\documentclass[aps,notitlepage, 10pt,onecolumn,english,showpacs,prd,longbibliography, citeautoscript,floatfix,superscriptaddress]{revtex4-1}
\usepackage{amsmath, mathrsfs}
\usepackage{amsthm}
\usepackage{amssymb}
\usepackage{amsfonts}
\usepackage{bbm}
\usepackage{epsfig}
\usepackage{times}
\usepackage{babel}
\usepackage{xcolor}
\usepackage{hyperref}
\usepackage{framed}

\hypersetup{
	colorlinks=true,
	linkcolor=blue,
	citecolor=blue,
	urlcolor=blue
}

\newcommand{\bra}[1]{\ensuremath{\left\langle#1\right|}}
\newcommand{\ket}[1]{\ensuremath{\left|#1\right\rangle}}
\newcommand{\braket}[2]{\ensuremath{\left\langle#1 \vphantom{#2}\middle|  #2 \vphantom{#1}\right\rangle}}
\newcommand{\ketbra}[2]{\ensuremath{\left|#1\right\rangle\!\left\langle#2\right|}}

\newcommand{\tr}[2]{\mathrm{Tr}_{#1}\left[ #2 \right]}
\newcommand{\iden}{\mathbb{I}}
\renewcommand{\v}[1]{\ensuremath{\boldsymbol #1}}

\newcommand{\R}{\mathbb{R}}
\newcommand{\be}{\begin{equation}}
\newcommand{\ee}{\end{equation}}

\newcommand{\E}{\mathcal{E}}

\newcommand{\I}{\mathbb{I}}
\newcommand{\rhomega}{\rho^{(\omega)}}
\newcommand{\om}[2]{(\omega_{#1 #2})}
\newcommand{\omp}[2]{(\omega_{#1'\! #2})}
\newcommand{\ompp}[2]{(\omega_{#1'\! #2 '\!})}

\theoremstyle{plain}
\newtheorem{thm}{Theorem}
\newtheorem{corol}[thm]{Corollary}
\newtheorem{lem}[thm]{Lemma}

\theoremstyle{definition}
\newtheorem{defn}{Definition}

\newtheorem{rmk}{Remark}

\newtheorem{ex}{Example}
\theoremstyle{remark}

\newtheorem*{centralq}{Central Question}

\begin{document}
	
	\title{An introductory review of the resource theory approach to thermodynamics}
	
	\author{Matteo Lostaglio}
	\address{ICFO-Institut de Ciencies Fotoniques, The Barcelona Institute of Science and Technology, Castelldefels (Barcelona), 08860, Spain}
	
	\begin{abstract}
		I give a self-contained introduction to the resource theory approach to quantum thermodynamics. I will introduce in an elementary manner the technical machinery necessary to unpack and prove the core statements of the theory. The topics covered include the so-called `many second laws of thermodynamics', thermo-majorisation and symmetry constraints on the evolution of quantum coherence. Among the elementary applications, I explicitly work out the bounds on deterministic work extraction and formation, discuss the complete solution of the theory for a single qubit and present the irreversibility of coherence transfers. The aim is to facilitate the task of those researchers interested in engaging and contributing to this topic, presenting scope and motivation of its core assumptions and discussing the relation between the resource theory and complementary approaches. 
	\end{abstract}
	
		\maketitle
	
		\tableofcontents

\bigskip

	This review is based on an introductory mini-course I gave in ICFO in May 2018 on the resource theory approach to quantum thermodynamics, itself partially based on my PhD thesis \cite{lostagliothesis}. Various quantum thermodynamics reviews have appeared in the last couple of years where the resource theory approach is touched upon; these include a broad review on the role of quantum information in thermodynamics \cite{goold2016role} (Section III); a broad review of the various approaches to quantum thermodynamics \cite{vinjanampathy2016quantum} (Section~5); and a short review stating and discussing some of the technical results of the resource theory approach~\cite{ng2018resource}.
	
	However, the aim of the present manuscript is different, in that I want to provide a self-contained introduction to the topic. Currently one has to delve into not always user-friendly appendices scattered throughout the literature to get a technical grasp of the subject; or alternatively accept certain statements without a solid understanding of the scope of the underlying assumptions and proof techniques. Especially in quantum thermodynamics, where a variety of diverse formalisms are being applied to closely related problems, uncritical acceptance does not allow to easily see the relations between complementary lines of research, as well as the current limitations of each approach; in turn, this hinders the development of the field as a whole. Contributing towards the solution of this issue is one of the aims of this manuscript. A second one is that, with hindsight, certain key proofs of the resource-theoretical framework (such as those involving thermo-majorisation) can be simplified considerably. Hence, I present a more direct derivation of some of the results appearing in the literature.
	
	Overall, I hope these notes will facilitate the work of those that wish to approach this topic and promote cross-talking within the larger quantum thermodynamics community and beyond. Great complementary tools are Markus M\"uller's lecture notes \cite{muellernotes} and some PhD thesis \cite{korzekwa2016coherence} (Part I), \cite{alhambra2017non, sparaciari2018multi}.
	
	With this premise, I would like to jump to the core of the matter. The manuscript is divided into three main sections:
	\begin{enumerate}
		\item Section I puts the resource theory approach to thermodynamics into context. This means showing that thermodynamics can be phrased in a language similar to entanglement and discussing how this approach relates to others.
		\item Section II deals with the energetic part of the theory. It proves that thermo-majorisation is the notion that characterises how different out of equilibrium distributions can be transformed into each other. I also discuss how thermo-majorisation is related to the absence of a unique extension of the concept of entropy (and free energy) to non-equilibrium states, and the so-called `second laws' of thermodynamics.
		\item Section III introduces the idea that thermodynamics can be seen as a theory of energy \emph{and} quantum coherence, to be understood as superpositions of different energy states. The study of the thermodynamic role of superpositions is naturally explored within the framework investigating symmetries of open dynamics; new, independent relations are necessary to capture the quantum aspects of the theory.
	\end{enumerate}
	The aim of being self-contained and relatively concise unavoidably clashes with completeness. Hence, many important results of the framework are not presented here. However, arguably thermo-majorisation and the theory of symmetry in open quantum systems are the two key tools upon which most further developments can be constructed. For this reason, I will focus mainly on these, while pointing the reader to other directions throughout the text. Unless otherwise stated, all systems will be assumed to be finite dimensional. We also set $\hbar =1$.
	
	\newpage \newpage 
	
	\section*{Introduction}
	
Despite its name, textbook thermodynamics \cite{callen1985thermodynamics} is, for the most part, not concerned with \emph{dynamics}. In fact, it can be rigorously formulated starting from the notion of ordering among equilibrium states \cite{giles2013mathematical, lieb1999physics}.

Within this general mindset, consider for concreteness  the following textbook thermodynamic question. We have a square box of volume $V$, containing an ideal gas at pressure $P$ and temperature $T$. One side of the box can be turned into a movable piston (perhaps removing a locking mechanism). The piston can be attached to a weight of known mass in a gravitational potential, initially at a height $L$. So, the initial state can be described as $X=(T,P,V,L)$. The weight may be used to do work on the gas by compressing it, or we may aim to raise the weight by allowing the gas to expand. Once things have settled down, the system will be described by new variables $X'=(T',P',V',L')$. A central question in thermodynamics is to work out what transitions are possible, i.e. what states are accessible from~$X$:
\be
\label{eq:intropossibletransitions}
X \stackrel{?}{\longrightarrow} X'.
\ee
If $X'$ is accessible from $X$, we will write $X \succ X'$. The relation $\succ$ of thermodynamic accessibility from one state to another is naturally endowed with the properties of a \emph{partial ordering}, once we group into equivalence classes $[X]$ all states $X$ that can be reversibly converted into each another (in the specific sense that both $X \succ X'$ and $X' \succ X$). In fact, $\succ$ is reflexive ($[X] \succ [X]$) and antisymmetric ($[X] \succ [X']$, $[X'] \succ [X]$ $\Rightarrow$ $[X'] = [X]$) by definition of the equivalence classes, if at least the trivial transformation $X \longrightarrow X$ is allowed; and it is transitive ($[X] \succ [X']$, $[X'] \succ [X'']$ $\Rightarrow$ $[X] \succ [X'']$) under the natural assumption that thermodynamic processes can be composed. The thermodynamic entropy (understood as a functional over the set of equilibrium states) is a tool to determine such ordering. In textbook thermodynamics an often implicit assumption is made: that one can construct a unique functional $S$ that characterises the allowed transitions:
\be
\label{eq:thermoentropyassumption}
X \succ X' \Leftrightarrow S(X)\leq S(X').
\ee
In general, $S$ that may be identified with the entropy or the free energy or a generalised thermodynamic potential, depending on the context.
 
What this takes for granted is that $\succ$ is a \emph{total ordering} among all equilibrium states. A partial ordering $\succ$ is called total when any two states are \emph{comparable}, i.e., given any two states  $X$ and $X'$, it is either $X \succ X'$ or $X' \succ X$. This is typically the case for transformations between equilibrium states and is referred to as Comparison Hypothesis~\cite{giles2013mathematical, lieb1999physics}. Characterising accessibility relations in non-equilibrium thermodynamics is much more complicated; we do not necessarily expect a total order and hence a unique extension of the concept of entropy to non-equilibrium states~\cite{lieb2013entropy}. 

Hopefully this discussion has hinted at what concepts need to be introduced and formalised:
\begin{enumerate}
	\item We need a precise notion of thermodynamic process. This will be captured within the framework of resource theories.
	\item We then need to characterise the (partial) ordering induced on the set of non-equilibrium states by the chosen notion of thermodynamic processes. In particular, we will need refined versions of the standard constraint ``entropy increases''.
\end{enumerate}
The second point naturally leads to the branch of mathematics that studies important notions of partial ordering. We will start with a quick introduction to resource theories. Before that, however, some clarifications concerning the scope of this construction are required.

\subsection*{Scope of the resource theory}

Having argued above for the usefulness of a resource theory approach to thermodynamics, it is crucial to stress that the overall aim is not to provide a supposedly final axiomatisation of the theory. Firstly, there are alternative approaches that are more convenient in many practical scenarios. For example, currently the resource theory does not deal easily with time-dependent Hamiltonians and does not incorporate quantum effects of macroscopic systems such as phase traditions; secondly, there is no general agreement concerning what the scope of quantum thermodynamics \emph{should be}, with many approaches trying to move beyond traditional boundaries.

On the one hand, the resource theory approach aims to go beyond the thermodynamic limit and the assumption of equilibrium; in fact, it is often presented as an extension of statistical mechanics to scenarios with potentially large fluctuations (so-called single-shot statistical mechanics \cite{dahlsten2011inadequacy, halpern2015introducing}). On the other hand, to the extent to which system and bath are small enough that their dynamics can be completely solved, we do not need anything more than quantum mechanics. Hence, one expects the resource theory approach to be most useful when a small number of quantum systems interacts with a large or perhaps finite-size but still intractable environment. This is the domain of the theory of open quantum system dynamics, so  it is not surprising that the two approaches are related. While the resource theory approach includes certain standard master equation treatments such as Davies maps, it does not rely on a master equation description and can in some sense be considered a non-Markovian generalization of these traditional treatments \cite{lostaglio2017markovian}. 

However, the main difference between the open system and the resource-theoretical approaches lie in the toolkits used and the questions asked. The resource theory allows us to describe and quantify certain thermodynamic properties from a quantum information perspective and, hopefully, gather new intuitions about the role out of equilibrium states play as resources to perform certain tasks. This goes beyond the question of solving a specific dynamics, and goes into the direction of describing the structure of non-equilibrium states, how their quantum properties may be harnessed, what is the optimal strategy for a given thermodynamic task and what are genuinely quantum advantages, if any.  It also aims to provide a suitable language to study the intersection between thermodynamics and information theory, and their quantum counterparts. We may draw a parallel with the resource-theoretical classification and manipulation of entanglement, which aids but by no means substitutes the study of entanglement in many body systems or investigations in the experimental generation and manipulation of highly entangled states. That is not to say that the resource theory does not have the potential to be relevant in applications, but the question is still far some settled (for some initial attempts, see Ref.~\cite{alhambra2018heat, halpern2018fundamental}, for some obstacles see Ref.~\cite{shu2018violation}).

Sitting in a grey area between various approaches, the resource theory is not an all-purpose framework to study generic `thermodynamic problems'. Rather, some of its central tools, most notably those used to characterize the partial order among states undergoing thermalisation, have the potential of being useful in complementary frameworks, e.g. by revealing common formal structures. A more conjectural possibility is that as quantum information tools are more widely applied in 'dramatic' scenarios, such as black hole physics and the AdS/CFT conjecture \cite{harlow2016jerusalem}, the quantum information approach to thermodynamics might provide valuable insights \cite{Bernamonti2018holographic}.

Not surprisingly, given their generality, some of the tools used in the resource theory approach were developed many years ago within the chemical physics community to study classical statistical systems (see the concept of `mixing distance' \cite{ruch1975diagram, ruch1976principle, ruch1978mixing}). Together with the study of symmetries of open quantum dynamics (whose study also has a long history, see e.g. \cite{holevo1993note, keyl1999optimal}), this provides the mathematical bedrock on which the resource theory relies. Due to their generality, these tools may have cascade applications into unconnected fields of the quantum sciences. This is another reason why the technical machinery of the resource theory is not hidden in an appendix, but on the contrary constitutes an integral part of these notes.

	\section{The resource theory of Thermal Operations}
	
	\subsection{The resource theory approach}
	
	The resource theory approach can be understood as the study of the limitations that arise due to some set of physical constraints. An infinitely powerful agent will experience no constraints, and consequently for her nothing will be a resource. For example, if she could teleport instantaneously anywhere in the Universe, neither time, nor a mean of transportation and the fuel necessary to operate it, would be a resource for the task of travelling. In the real world, however, we face both fundamental constraints (e.g., the speed of light is finite) and practical constraints (e.g., limited carbon budget). Given a set of constraints, we can then identify a corresponding set of resources (bus tickets, fuel, etc.) which may be consumed to perform a given task (go from point $A$ to point $B$).
	
	In the context of quantum theory, agents are limited in their ability to manipulate a quantum state by the laws of quantum mechanics. More than that, extra constraints may arise from the specification of the physical setting, as this well-known example shows \cite{horodecki2003entanglement}:
	
	\begin{ex}[LOCC transformations]
		\label{ex:locc}
		Imagine Alice and Bob are at two different locations $A$ and $B$. Within each laboratory, each of them is allowed to perform any quantum operation, represented by an arbitrary Completely Positive and Trace Preserving (CPTP) linear map, also known as channel. Furthermore, Alice and Bob can communicate, but only classically: for example, Bob can condition the operation performed on his side on the outcome of a measurement on Alice's side. The set of operations identified in this way is known as Local Operations and Classical Communication (LOCC). 
	\end{ex}
	More generally, the choice of a subset of all quantum channels defines a set of \emph{allowed transformations}; we assume an agent can perform such operations in arbitrary number (weakening this assumption may be an interesting line of research). We will not delve into details about general resource theories, which is an interesting research topic since it seeks to identify common features of generically vastly different operational frameworks (the interested reader can consult, e.g., \cite{horodecki2013quantumness, brandao2015reversible, gour2017quantum, chitambar2018quantum} and references therein). From an operational point of view, it often makes sense that the set of allowed operations is
	\begin{enumerate}
		\item \emph{Convex}: if one is allowed to build a box that performs operation $1$ and a box that performs operations $2$, it is natural to assume that we can build a third box that implements $1$ or $2$ conditioned on the outcome of a coin toss.
		\item \emph{Closed under composition}: if one can perform operations $1$ and $2$, one can also implement them in a sequence.
		\item \emph{Contains the identity channel}: doing nothing is allowed. 
	\end{enumerate}
	
	The choice of a set of allowed transformations goes hand in hand with the identification of a set of states that can be generated by allowed operations only. For example, Alice and Bob can produce any separable state by means of LOCC transformations. Such quantum states are called  \emph{free states} of the theory. Any state that is \emph{not} a free state will be called a \emph{resource state}. For the set of LOCC transformations, resource states are entangled states. A set of allowed operations with its correspondent set of free states defines a \emph{resource theory}. The previous example defines the resource theory of entanglement. In general:
\begin{defn}
	A \emph{resource theory} is defined by a subset of all quantum channels (\emph{allowed operations} $\mathcal{A}$), closed under composition and including the identity, and by the set of all states that can be generated by allowed operations only (\emph{free states} $\mathcal{F}$). Any state that is not free is a \emph{resource state}, denoted by $\mathcal{R}$. 
\end{defn}

If we can only access free states, things are boring since we can explore the full set $\mathcal{F}$ and never get out of it. However, if we are given some $\rho \in \mathcal{R}$, we can ask how this resource can be used. For example, we might want to know if it can be manipulated to another form $\sigma$ by operations in $\mathcal{A}$. We will write
\begin{equation}
\rho \succ_{\mathcal{A}} \sigma \quad \Leftrightarrow \quad \exists A \in \mathcal{A}: A(\rho) = \sigma.
\end{equation}
It should be clear that if $\rho \succ_{\mathcal{A}} \sigma$, then $\rho$ is a better resource than $\sigma$ (strictly better if  $\sigma \nsucc_{\mathcal{A}}  \rho$). An important aspect of resource theories is that they define a \emph{partial ordering} on the set of quantum states, defined by the above notion of convertibility through allowed operations:
\begin{rmk}[Partial order of resources]
	\label{ex:partialorder}
	Denote by $[\rho]$ the equivalence class of states that can be \emph{reversibly interconverted} into each other, $[\rho] = \{\sigma|\; \rho \succ_{\mathcal{A}} \sigma \textrm{ and } \sigma \succ_{\mathcal{A}} \rho  \}$. The relation $\succ_{\mathcal{A}}$ defines a partial order among these classes of states. 
\end{rmk}
	 While many physical questions can be ultimately understood mathematically as the description of the partial ordering among states in $\mathcal{R}$, there is a different and useful point of view on resources. If we are given $\rho \in \mathcal{R}$, we may `consume' it to simulate a non allowed operation, matching the intuition of resources as `fuel' to realise otherwise impossible operations.
	 \begin{ex}[Quantum teleportation]
	 	Let $\mathcal{A} =$ LOCC and $\mathcal{F}=$ separable states. Suppose we are given a copy of the Bell state \mbox{$\ket{\phi^+} = (\ket{00}+\ket{11})/\sqrt{2}$}. Then by allowed operations we can realize one application of the identity channel from $A$ to $B$, a channel that is not $\mathcal{A}$. This is the well-known quantum teleportation protocol \cite{nielsen2010quantum}. In performing the protocol, we consume the resource.
	 \end{ex}

	  	  \subsection{The intuition behind Thermal Operations}
	  	  
	  	  A resource theory of quantum states out of equilibrium aims to answer the question: given some initial state and a set of thermodynamically allowed transformations, what is the set of reachable final states? This approach to thermodynamics is, in many ways, inspired by the theory of entanglement. Entanglement theory allowed us to understand that entanglement comes in different `flavours' and that certain resources of particular interest (like $\ket{\phi^+}$) can be distilled from many copies of `worse' resource states by means of allowed operations \cite{horodecki2003entanglement}. The above aspects will have thermodynamic analogues in the `distillation' of pure energy states and in a complicated partial order among non-equilibrium states, with more familiar relations emerging only in specific limits. But how does one define a set of thermodynamically free operations?
	  	 	  
	  	  Various definitions have been proposed. These try to formalise the notion of operations that can be carried out at no cost.  It is important to stress that ultimately the definition of this set is a postulate of the theory -- and in any given circumstance the best choice will arguably depend on the experimental limitations we want to describe. Here we will now show how to justify heuristically the most common choice -- Thermal Operations -- from so-called passivity considerations, but in the next subsection we will also mention some of the alternatives.
	  	  
	  	  \subsubsection*{Why only thermal states can be free states?}
	  
	  A system in isolation is given, with Hamiltonian $H_S$. We ask: what states do not constitute a resource? A \emph{necessary} condition should be that such states can be prepared at no work cost, but this is necessarily relative to a background reference temperature. That is, the same system may be or not a resource depending on the environment in which it is embedded (e.g. a thermal state at $300K$ at the North Pole will give rise to a heat flow that can be used to extract work, whereas in a room at $300K$ it will not). As we now argue, if the system $X$ has Hamiltonian $H_X$, only the state $\gamma_X = e^{-H_X/kT}/\tr{}{e^{-H_X/kT}}$ can potentially be a free state in the theory, where $T$ is the temperature of the background environment, and $k$ is Boltzmann's constant. In the following we will use the notion of inverse temperature $\beta = (kT)^{-1}$ and of partition function $Z_S = \tr{}{e^{-H_S/kT}}$. 
	  
	  Why can we argue that any state different from $\gamma_X$ cannot be prepared for free? Suppose that some other state $\rho_X \neq \gamma_X$ can be prepared for free. Then also $\rho^{\otimes n}_X$ can be prepared, for an arbitrary $n$.  
	  But one can show that for $n$ large enough there exists a unitary $U$ such that
	  \begin{equation}
	  \label{eq:workextractionpassive}
	  \tr{}{U \rho^{\otimes n}_X U^\dag H_n} < \tr{}{ \rho^{\otimes n}_X H_n}, 
	  \end{equation}
	  where $H_n = \sum_i H_{X,i}$, $H_{X,i} = \iden_1 \otimes \dots \otimes \iden_{i-1} \otimes H_X \otimes \iden_{i+1} \otimes \dots \otimes \I_n$ \cite{lenard1978thermodynamical}. $U$ can be understood as being the unitary induced by a time dependent Hamiltonian that starts in $H_n$ at the beginning of the protocol and returns to $H_n$ at the end. States whose energy can be lowered by unitaries are called \emph{active} and the above result says that if $\rho_X \neq \gamma_X$, then $\rho_X^{\otimes n}$ is active for some $n$. It is generally accepted that active states are those from which one can extract work, since from the first law of thermodynamics the change of energy equals work when entropy is constant. 
	  Hence, the heuristic goes, the preparation of $\rho_X \neq \gamma_X$ requires some work and should not be deemed a free operation.  	  
	  
	  If $\gamma_X$ is free, then one can create states such as $\gamma_S \otimes \gamma_B$, where $\gamma_S \propto e^{-H_S/kT}$,  $\gamma_B \propto e^{-H_B/kT'}$. However one must have $T'=T$, otherwise $(\gamma_S \otimes \gamma_B)^{\otimes n}$ would be active on $SB$ for some~$n$ (put it more physically, we could operate a Carnot engine between the two temperatures and extract work). $T$ can be identified with a fixed background temperature.
	  
	    \begin{rmk}[Refinements of the passivity argument]
	  	Some readers may have noticed that in the above passivity argument we allowed to perform general unitaries to extract work from $\rho_X$; in fact, we could have restricted to unitaries that commute with the sum of the local Hamiltonians (`energy-preserving unitaries'), and required to model the work extraction from $X$ by introducing an explicit battery system. Such discussion can be found in Appendix D of Ref.~\cite{brandao2013second}, where it is shown that given enough copies of $\rho_X \neq \gamma_X$ one can increase the average energy of a harmonic oscillator, initially prepared in an eigenstate. We then reach the same conclusion.
	  	 	 However, note that the theory is not necessarily fully trivialised (i.e., all transitions become possible) if one allows non-thermal states for free: e.g. one may still be unable to create quantum coherence in the energy basis if arbitrary diagonal states and energy-preserving unitaries are given (we will come back to this in Sec.~\ref{sec:decomposition}). Hence, the requirement that the theory is non-trivial does not uniquely select the thermal state as the only free state.

	  		Finally note that one might allow the use of active ancillary states, as long as these are given back unchanged at the end of the process (so that no work is extracted from them). We will discuss this possibility in Sec.~\ref{sec:thermodynamicschur}, but to avoid making the current heuristic argument excessively complicated we will not delve into that here.
	  		  	
	  	\end{rmk}
	  	
	   \begin{rmk}[Multiple temperatures] In thermodynamics, we are often interested in scenarios involving baths at different temperatures. One way to deal with this is to work within a resource theory where the background temperature is $T$, and simply treat a second bath at temperature $T'$ as a resource; an alternative is to consider two resource theories and alternate between them.  Arguably, however, finding natural ways to bookkeep resources in the presence of multiple temperatures is an interesting direction in which to extend the framework. For some results in this direction, see Ref.~\cite{sparaciari2017resource}.
	  \end{rmk}
	  
	    \subsubsection*{Why only energy-preserving unitaries are free?}
	    
	  Let us now discuss evolutions. From quantum theory we know that isolated systems evolve unitarily, and we can assume that the composite system $SB$ (system+environment) is isolated or nearly so, once all relevant systems have been included. Consider the action of such unitary $U$ on $\gamma_S \otimes \gamma_B$. One can note that if $U \gamma_S \otimes \gamma_B U^\dag \neq \gamma_S \otimes \gamma_B$, then the unitary can generate many copies of a non thermal state by acting on $(\gamma_S \otimes \gamma_B)^{\otimes n}$, hence preparing an active state when $n$ is large enough. To prevent this, we require that $U \gamma_S \otimes \gamma_B U^\dag = \gamma_S \otimes \gamma_B$, which can be equivalently written as $ e^{-\beta U(H_S + H_B)U^\dag} = e^{-\beta (H_S + H_B)}$, which implies  $[U,H_S + H_B] = 0$ (here we assumed $\gamma_S \otimes \gamma_B$ to be full rank; also, the eigenvalues of $H_S+H_B$ can be taken wlog all non-negative, so there exists a principal logarithm of $\gamma_S \otimes \gamma_B$). This condition can also be understood as microscopic energy conservation on system plus environment, hence $U$ is called an \emph{energy-preserving unitary}.
	  
	  This relation is sometimes cause of confusion, since in thermodynamics we are used to drive systems through time-dependent Hamiltonians and the above commutation relation typically does not hold. In some cases, this is simply because driving does cost net work, which comes from an implicit classical system, and so should not be a free operation. Furthermore, in some cases one only takes into account the \emph{average} energy cost of the driving, so that in some setups it is natural to assume average energy conservation rather than the stricter $[U,H_S + H_B] = 0$ \cite{skrzypczyk2014work}. Going even further one could consider what happens in the presence of generic violations of the energy condition, but it is then unclear what kind of statements -- if any -- can be inferred from general considerations alone, i.e. without solving the dynamics \cite{shu2018violation}. However, here we will be working within a fully quantised picture, in which all energy fluctuations will be explicitly accounted for by the introduction of quantum systems that play the role of the work repositories. This becomes relevant for small systems, or in machines in which we do not average over a large number of cycles. 
	  
	  Another issue is that the relation $[U,H_S+H_B] = 0$ seems incompatible with strong coupling scenarios. However, the situation is a bit more subtle here. First, the above relation is compatible with some strong coupling scenarios (think of scattering events, for example, identifying $U$ with the scattering matrix  and singling out incoming and outgoing subspaces \cite{janzing2005decomposition}). Second, the relation may be deemed incompatible with situations in which the system remains strongly coupled to the environment at the end of the transformation, or $S$ and $B$ are taken to be strongly coupled at all times. However, 
	   in this situation one can argue that our notion of thermodynamic system $S$ should be adapted accordingly, to include the strongly interacting part of the environment. What is certain is that more work is needed to formalise this procedure in the resource theory context. We will come back to this in the next subsection.

	  \subsection{Definition of the set of Thermal Operations, extensions and restrictions}
	  
	  The previous heuristics lead us to identify the free states with thermal states at the environment temperature, since from any other state one could extract work with energy preserving unitaries (given enough copies). Furthermore, any unitary outside the class of energy preserving unitaries allows to extract work from thermal states. This suggests the following set of allowed operations \cite{janzing2000thermodynamic, brandao2011resource}:
	  \begin{enumerate}
	  	\item Preparing thermal states at a fixed temperature $T$ and with arbitrary Hamiltonian $H_B$,
	  	\item Performing energy-preserving unitaries,
	  	\item Tracing out subsystems.
	  \end{enumerate}
  These can be combined to generate a channel $\mathcal{T}$ on the system $S$ with Hamiltonian $H_S$:
  \begin{equation}
  \label{eq:thermal}
  \mathcal{T}(\rho_S) = \tr{B}{U(\rho_S \otimes \gamma_B)U^\dag},
  \end{equation}
  where $\gamma_B = e^{-\beta H_B} / Z_B$, $Z_B = \tr{}{e^{-\beta H_B}}$, $H_B$ is arbitrary and $U$ is any unitary satisfying $[U, H_S + H_B] =0$. This set of channels are known as Thermal Operations; they are a convex set, as one can verify by introducing appropriate ancillas (Appendix~C of \cite{lostaglio2015quantum}). One can check directly from the definition that the thermal state $\gamma_S =  e^{-\beta H_S} / Z_S$ is a fixed point of every Thermal Operation, i.e. $\mathcal{T}(\gamma_S) = \gamma_S$. 
  A central question of the theory is the following: 
  \begin{centralq}
  Given $\rho_S$, $\sigma_S$, is there a Thermal Operation $\mathcal{T}$ such that $\mathcal{T}(\rho_S) = \sigma_S$?
  \end{centralq}

  While we will focus on Thermal Operations, some remarks are in order. First, one can generate channels with different input and output spaces, by tracing over different degrees of freedom (rather than $B$); second, transformations involving changes in the system Hamiltonian can be introduced by means of an auxiliary system with trivial Hamiltonian that acts as a switch; for example, set $S \equiv S' \otimes T$ ($S'$ being the new system, and $T$ a switch), $H_S = H_{S'}(0) \otimes \ketbra{0}{0}_{ T} + H_{ S'}(1) \otimes \ketbra{1}{1}_{T}$ and restrict to the study of transitions of the form $\rho_{S'} \otimes  \ketbra{0}{0}_{T} \rightarrow \sigma_{S'} \otimes  \ketbra{1}{1}_{T}$. This implicitly encodes the fact that the system Hamiltonian changes from $H_{S'}(0)$ to $H_{S'}(1)$ during the transformation (Appendix H of Ref.~\cite{brandao2011resource}). Another way to encode this, successfully employed to model fluctuation theorems, is to assume, on top of $[U,H_S + H_B]=0$, that $U (\mathbb{I}_{S'} \otimes \ketbra{0}{0}_T \otimes \mathbb{I}_B) = (\mathbb{I}_{S'} \otimes \ketbra{1}{1}_T \otimes \mathbb{I}_B) U$ (Appendix A of \cite{aberg2016fully}). When we put these considerations together, we see that in principle one can deal with situations in which some interaction Hamiltonian $H_{SB}$ is switched on during the process and remains on at the end; then, since we cannot simply trace away a strongly interacting part of the environment, one needs to choose the output space of the channel in a way that includes part of the environment (e.g., by reaction coordinates techniques \cite{nazir2018reaction}). Yet another way to deal with this situation is to generalise the commutation relation $[U,H_S+H_B] = 0$ to $U (H_S + H_{B}) = (H_{S'} + H_{B'}) U$, where $SB$ and $S'B'$ are two different factorizations of the total Hilbert space. This takes into account the fact that our notion of the splitting between system and environment might change during the process.
  
  These are subtleties that is worth mentioning, but we will not discuss them further in this introduction (really, these questions are common to all approaches to small-scale thermodynamics). A complete clarification of these points would be desirable, including the relation between Thermal Operations and the standard approximations adopted in the open quantum systems approach to quantum thermodynamics \cite{davies1974markovian, kosloff2013quantum}. For some partial results, see Section~E of Ref.~\cite{lostaglio2017markovian} and Remark~\ref{rmk:masterequations}.
  
   It is worth making some concluding foundational remarks. Performing operations on a system as mandated by the resource theory framework implicitly assumes the existence of external, classical degrees of freedom out of equilibrium that act as a clock to measure durations. This entails potential extra costs to run the clock~\cite{erker2017autonomous} or the construction of a machine that avoid its use \cite{clivaz2019unifying}. However also note that one might expect on general grounds that such clock systems are dynamically generated in subsystems of the Universe in which thermodynamic laws emerge~\cite{barbour2014identification}.    
  
  \begin{rmk} [An alternative intuition for the thermodynamic resource theory]
  	\label{rmk:alternative}
  	In the previous section we introduced Thermal Operations through a passivity-type argument. An alternative route to a thermodynamic resource theory is the following. As discussed in more detail later (Sec.~\ref{sec:timetranslationsymmetryTO}), Thermal Operations are channels that satisfy two central properties:
  	\begin{enumerate}
  		\item $\mathcal{T}(\gamma_S) = \gamma_S$ (thermal fixed point).
  		\item $\mathcal{T}(e^{-i H_S t} \rho_S e^{i H_S t}) = e^{-i H_S t} \mathcal{T}(\rho_S) e^{i H_S t}$ for all $\rho_S$ (symmetry under time translations).
  	\end{enumerate}
  	The first property is the stability of the thermal state under operations that use no work. Its physical meaning is relatively straightforward. The second property can be shown to be equivalent to the principle that energetic coherence cannot be freely created. That is, for every $\rho_{SA}$,
  	\begin{equation}
  	[\rho_{SA}, H_S + H_A] = 0 \Rightarrow  [\mathcal{T} \otimes \mathcal{I}_A (\rho_{SA}), H_S + H_A]=0,
  	\end{equation}
  	where $\mathcal{I}_A$ is the identity operation on the ancilla (this follows from the structure of the Choi-Jamio\l kowski state of covariant channels, see Appendix A3 of Ref.~\cite{lostaglio2017markovian}). If one accepts the intuition behind these two principles (thermal fixed point and no free creation of energetic coherence), a resource theory can be defined from the set of all channels compatible with them. These channels are known as Thermal Processes or Enhanced Thermal Operations~\cite{cwiklinski2015limitations, gour2017quantum2}. While it is not clear if the theory of Thermal Processes and that of Thermal Operations are exactly equivalent (see conjecture~\ref{conjecture}), all the considerations of these notes apply equally well to both.
  	
  \end{rmk}
  \begin{rmk}
  	\label{rmk:extensionslimitations}
  	[Extensions and restrictions of Thermal Operations].  As with the theory of entanglement, there are larger and smaller sets one may consider, for physical or technical reasons. 
  	
  	Larger sets may be considered for their simpler structure (we have already mentioned Thermal Processes), or as alternative frameworks in which some extra resources are allowed (e.g. Gibbs-preserving channels, i.e. the set of all channels with a thermal fixed point \cite{faist2015gibbs}, or Catalytic Thermal Operations, which allow the use of auxiliary systems that are given back unchanged at the end, see Sec.~\ref{sec:thermodynamicschur}). Another interesting extension is the abovementioned weakening of energy conservation to average energy conservation. This choice leads to a framework where one can perform any transformation that decreases the non-equilibrium quantum free energy $F(\rho_S) = \tr{}{\rho_S H_S} - kT S(\rho_S)$, where $S$ is the von Neumann entropy, $T$ the temperature of the environment and $k$ Boltzmann's constant (this can be inferred from Ref.~\cite{skrzypczyk2014work}, but also derived without their use of a weight system). The main drawback is that fluctuations in the energy source are, by definition, ignored. Depending on the task considered, such fluctuations may be important in the functioning of small devices.  More general violations of the condition of energy conservation should, in general, be properly accounted for. In the resource theory this is often done by the introduction of explicit battery systems, studying the action of free operations on the composite system+battery (see Sec.~\ref{sec:workextractionincoherent}). Note, however, that one can depart from this, if appropriate and relevant for the given setup. For example, in cooling it is control, rather than work expenditure, the main limiting factor and cooling rates/cooling per round are relevant metrics to evaluate the performance of a protocol. In this setting it is then natural to allow non energy-preserving unitaries on the system~\cite{alhambra2018heat, clivaz2019unifying}. Creative uses of the theory presented in these notes, and their combination with other frameworks, will be important in applications.
  	
  	
  	Smaller sets are also often considered, in order to understand what extra limitations arise from further, and hopefully more realistic, physical constraints (e.g. Elementary Thermal Operations \cite{lostaglio2018elementary, mazurek2018decomposability}, or Thermal Operations with a single bosonic mode \cite{hu2019thermal}), or to study if simpler subsets of operations still allow to achieve the same transformations as Thermal Operations (e.g. Coarse Thermal Operations \cite{perry2015sufficient}). One can also notice that the standard open quantum system description of weak coupling interactions with a large thermal environment, developed by Davies \cite{davies1974markovian} and formally defined in Ref.~\cite{roga2010davies}, is included in the set of Thermal Processes; that is, we can understand Thermal Processes as a non-Markovian extension of Davies maps. Hence, another relevant restriction to consider is the Markovian subset of Thermal Operations/Thermal Processes \cite{lostaglio2017markovian}.
  	  
  	We will mostly focus here on Thermal Operations, but many core techniques can be applied to alternative setups as well. Furthermore, extension to multiple conserved quantities are generally straightforward \cite{halpern2016beyond}, even though subtleties arise in the presence of mutually non-commuting conserved quantities \cite{balian1987equiprobability, lostaglio2017thermodynamic, halpern2016microcanonical,guryanova2016thermodynamics}. Finally, it would be interesting to develop a complete resource theoretical approach to moving between different descriptions of a thermodynamic system -- as given by agents with different levels of control over it \cite{jaynes1965gibbs} (see Section~3D of Ref.~\cite{faist2018fundamental} for some results in this direction).
  \end{rmk}

\section{Thermodynamic laws for population}
  
  While the second law is often expressed as a principle informing us that certain processes are \emph{impossible}, following the foundational works of Giles, Lieb and Yngvason we more ambitiously ask: what transformations are \emph{possible}? The thermodynamic laws define constraints encoding the partial order on the set of quantum states under the set of thermodynamically allowed operations.
  Hence, the characterisation of such partial order is the main technical problem we face. We will begin by considering what transformations on the occupations of different energy levels are possible. 
  
  \begin{rmk}[What about heat?] Note that we will focus on the existence of an environment that transforms the system from some given initial state to some target final state. Hence, less focus is put on the heat flows to and from the bath than standard treatments. For a resource theory approach with no free states in which, so to speak, heat is not free, see Ref.~\cite{sparaciari2017resource}.
  	 \end{rmk}

\subsection{Characterisation of thermal stochastic processes} 
 Unless otherwise mentioned, for simplicity we will assume throughout this introduction that $H_S$ is non degenerate, so that \mbox{$H_S = \sum_i E^S_i \ketbra{i}{i}$} with $E^S_1 < E^S_2 < \dots < E^S_{n}$. However, most statements can be trivially extended to Hamiltonians with degeneracies. Note that we will use the term \emph{population} or \emph{occupations} to indicate, given a general quantum state $\rho_S$, the vector $\v{p}$ whose elements are $p_i = \bra{i}\rho_S\ket{i}$.

As a necessary step to characterise the thermal partial order, we need a better way to describe the set of allowed transformations: as it is, these are very implicitly defined by invoking arbitrary energy-preserving unitaries and Hamiltonians of the environment, see Eq.~\eqref{eq:thermal}. The next result gives a much more economic description. We will need to define a set of stochastic matrices called $\emph{Gibbs-stochastic}$. First, recall that a \emph{stochastic matrix} is simply a matrix $G$ of transition probabilities $G_{i|j}$, i.e. $G_{i|j} \in [0,1]$ and $\sum_i G_{i|j} =1$ for all $j$. Now, $G$ is called Gibbs-stochastic if it preserves the Gibbs distribution $\v{g}$, $g_i = e^{-\beta E^S_i}/Z_S$, with $\beta = (kT)^{-1}$. In other words, $G \v{g} = \v{g}$. The central role of Gibbs-stochastic matrices follows from the following theorem \cite{horodecki2013fundamental, korzekwa2016coherence}. Informally, it states that the population dynamics generated by Thermal Operations are all and only the Gibbs-stochastic matrices. Formally,

\begin{thm}[Action of Thermal Operations on population]
	\label{th:thermalgibbs}
	Let $\rho_S$ and $\rho'_S$ denote two quantum states, with corresponding population vectors $\v{p}$ and $\v{p}'$. If there is a Thermal Operation $\mathcal{T}$ such that $\rho'_S = \mathcal{T}(\rho_S)$, then the population vectors are related as follows:
	\begin{equation}
\v{p}'=	G^{\mathcal{T}} \v{p},
	\end{equation}
	where $G^\mathcal{T}$ is the Gibbs-stochastic matrix $G^{\mathcal{T}}_{i|j} = \bra{i}\mathcal{T}(\ketbra{j}{j})\ket{i}$. Conversely, for every $\epsilon > 0$ and Gibbs-stochastic matrix $G$, there exists a Thermal Operation $\mathcal{T}$ that acts on the population as $G^{\mathcal{T}}$ and satisfies \mbox{$\max _{i,j} |G_{i|j} - G^{\mathcal{T}}_{i|j}|\leq \epsilon$}.
\end{thm}
\begin{proof}
	Let $\mathcal{D}$ be the quantum map that removes every off-diagonal element of a quantum state in the energy eigenbasis, called \emph{dephasing} map. Note that
	\begin{equation}
	\label{eq:dephasing}
	\mathcal{D}(\rho_S) = \frac{1}{s} \int_{0}^{s} e^{-i H_S t} \rho_S e^{iH_S t}dt,
	\end{equation}
	for an appropriately large $s$ (potentially one needs to take the limit $s \rightarrow \infty$). Then, from the general expression in Eq.~\eqref{eq:thermal}, using $e^{i H_B t} \gamma_B e^{-i H_B t} = \gamma_B$ and $[U, H_S + H_B] = 0$, one can verify $\mathcal{D} \circ \mathcal{T} = \mathcal{T} \circ \mathcal{D}$. This relation provides $\mathcal{D}(\rho'_S) =  \mathcal{T}[\mathcal{D}(\rho_S)]$. $\mathcal{D}(\rho_S)$ ($\mathcal{D}(\rho'_S)$) is a diagonal matrix whose nonzero elements are $\v{p}$ ($\v{p}'$). Note that here we are using that $H_S$ is non-degenerate, but if there was any degeneracy we could transform the dephased state into a diagonal matrix by unitaries commuting with $H_S$, which are Thermal Operations. Then, using $p_i = \bra{i} \mathcal{D}(\rho_S) \ket{i}$, $p'_j = \bra{j} \mathcal{D}(\rho'_S) \ket{j}$,
	\begin{equation}
	p'_j = \bra{j}\mathcal{T} \left(\sum_i p_i \ketbra{i}{i}\right)\ket{j} = \sum_i p_i \bra{j}\mathcal{T} (\ketbra{i}{i})\ket{j} = \sum_i  G^{\mathcal{T}}_{j|i} p_i.
	\end{equation}
	  Since $\mathcal{T}$ is trace-preserving and positive, $G^\mathcal{T}$ is stochastic. Furthermore, we can verify from Eq.~\eqref{eq:thermal} that $\mathcal{T}(\gamma_S) = \gamma_S$. Since $\gamma_S$ is a diagonal matrix with elements $\v{g}$, we have $G^{\mathcal{T}} \v{g} = \v{g}$, i.e. $G^\mathcal{T}$ is Gibbs-stochastic. This concludes the first part of the proof.
	
	The converse is based on an explicit construction of a Thermal Operation. Let $G$ be a target Gibbs-preserving matrix. Since for any Thermal Operation $\mathcal{T}$ the induced Gibbs preserving matrix $G^\mathcal{T}$ only depends on the action of $\mathcal{T}$ on the diagonal elements of $\rho_S$,
	without loss of generality take $\rho_S$ of the form
	\begin{equation}
	\rho_S = \sum_{i} p_i \ketbra{i}{i}.
	\end{equation} 
	Taking $H_B = \sum_j \sum_{g=1}^{g(E^B_j)} E^B_j\ketbra{E^B_j}{E^B_j}$, where $g(E_j^B)$ is the degeneracy of energy $E_j^B$, the environment state will be
	\begin{equation}
	\gamma_B = \frac{1}{Z_B} \sum_j e^{-\beta E^B_j} \sum_{g=1}^{g(E^B_j)} \ketbra{E^B_j,g}{E^B_j,g}.
	\end{equation}
Then, setting $E = E^S_i + E^B_j$ and summing over $E$ and $E_i^S$ rather than $E_i^S$ and $E_j^B$
	\begin{equation}
	\rho_S \otimes \gamma_B = \sum_E \sum_i \sum_{g=1}^{g(E-E^S_i)} p_i \frac{e^{-\beta E}}{Z_B}e^{\beta E^S_i} \ketbra{E,i,g}{E,i,g}.
	\end{equation}
	Now \emph{assume} exponential degeneracy $g(E-E^S_i) = g(E) e^{-\beta E^S_i}$. Then, setting $d_i(E) := g(E)e^{-\beta E^S_i}$,
		\begin{equation}
	\rho_S \otimes \gamma_B = \sum_E \frac{e^{- \beta E}}{Z_B}g(E) \ketbra{E}{E} \otimes \sum_i \sum_{g=1}^{d_i(E)} \frac{p_i}{d_i(E)} \ketbra{i,g}{i,g}.
	\end{equation}
 Note that Thermal Operations allow any energy-preserving unitary on every subspace of constant energy $E$, i.e. \mbox{$U = \oplus U_E$} with $U_E$ arbitrary unitary on the subspace of energy $E$. In particular,  $\sum_i \sum_{g=1}^{d_i(E)} \frac{p_i}{d_i(E)} \ketbra{i,g}{i,g}$ is a quantum state and has blocks of $d_i(E)$ copies of each eigenvalue $p_i/d_i(E)$. We can choose $U_E$ to be any permutation of these. If $n_{i|j}(E)$ is the number of eigenvalues transferred from block $j$ to block $i$, we have the conditions $\sum_i n_{i|j}(E) = d_j(E)$ and $\sum_j n_{i|j}(E) = d_i(E)$. The state after the unitary is then
	\begin{equation}
	 \sum_E \frac{e^{- \beta E}}{Z_B}g(E) \ketbra{E}{E} \otimes \sum_i \sum_j \frac{n_{i|j}(E)}{d_j(E)} p_j \ketbra{i,g}{i,g}.
	\end{equation}
	Let $G^{\mathcal{T}}_{i|j}(E) = \frac{n_{i|j}(E)}{d_j(E)}$ be the matrix acting on the population $\v{p}$. From the conditions on $n_{i|j}(E)$ one can verify that $G^{\mathcal{T}}(E)$ is Gibbs-stochastic: $\sum_i G^{\mathcal{T}}_{i|j}(E)  =1$, $\sum_j G^{\mathcal{T}}_{i|j}(E)e^{-\beta E_j^S} = \sum_j n_{i|j}(E)/g(E) = e^{-\beta E^S_i}$. Furthermore, by taking $g(E)$ large enough we can achieve any rational approximation of any set of transition probabilities giving rise to a Gibbs-stochastic matrix. We conclude that $G^{\mathcal{T}}(E)$ can be an arbitrary Gibbs-stochastic matrix, up to an arbitrarily small error $\epsilon$. By appropriate choices of $d_i(E)$ and $n_{i|j}(E)$ for each $E$, we can make sure the same Gibbs-stochastic matrix is applied within every block of total energy $E$. Hence, we conclude that an arbitrary Gibbs-stochastic matrix $G$ can be arbitrarily well approximated by $G^{\mathcal{T}}$.  
	
\end{proof}

Thanks to the previous theorem, we can answer the question of the existence of a Thermal Operation mapping the population $\v{p}$ of $\rho_S$ into the population $\v{p}'$ of $\rho'_S$ by tackling the question: when is there a Gibbs-stochastic matrix $G$ mapping $\v{p}$ into $\v{p}'$? Note that Gibbs-stochastic maps are only approximated arbitrarily well by Thermal Operations, but this is not of concern; it simply means that if there is $G$ Gibbs-stochastic such that $G \v{p} = \v{p}'$, then for every $\epsilon>0$ there is a Thermal Operation $\mathcal{T}_\epsilon$ such that, if $\v{p}'_\epsilon$ is the population of $\mathcal{T}_\epsilon (\rho_S)$, one has $\| \v{p}'_\epsilon - \v{p}'\| \leq \epsilon$. $\epsilon$-closeness is related to the indistinguishability of the distributions through the notion of total variation distance, that is extended to quantum states by the Holevo-Helstrom theorem \cite{watrous2018theory}. One could also allow for a \emph{fixed} (as opposed to arbitrarily small) $\epsilon >0$, which corresponds to a coarse-graining that extends the set of possible transformations \cite{horodecki2013fundamental, aberg2013truly}.   The bottom line of the above theorem is, in physical terms: \emph{Thermal Operations are equivalent to Gibbs-stochastic matrices if one only looks at the diagonal}.

 \begin{rmk}[What baths are allowed by Thermal Operations?]

The proof of the previous theorem clarifies that we are allowing complete freedom in the choice of the bath. In particular, we have chosen an exponential degeneracy in the states. Physically, this assumption comes from an assumption of infinite heat capacity, which itself can be understood as assuming that the heat bath can have infinite volume (effectively this may require an interaction that lasts infinitely long). This can be understood from the following classical reasoning, that we keep heuristic. Note that, by the standard definition of entropy and setting $k=1$, 
\begin{equation}
g(E + \epsilon) = e^{S(E+\epsilon)} = e^{S(E) + \epsilon \frac{\partial S}{\partial E} + \frac{\epsilon^2}{2} \frac{\partial^2 S}{\partial E^2} + O(\epsilon^3)},
\end{equation}
Now, one defines $\beta = \frac{\partial S}{\partial E}$ and hence $\frac{\partial^2 S}{\partial E^2} = -\frac{1}{T^2} \frac{\partial T}{\partial E}:= -\frac{1}{C T^2}$, where $C$ is the heat capacity, $C = \frac{\partial E}{\partial T}$. Then, $g(E + \epsilon) = g(E)e^{\beta \epsilon} e^{-\epsilon^2/(2C T^2) + o(\epsilon^2)}$. So we see that in the proof we are effectively allowing a bath with infinite heat capacity or, since $C \propto V$ ($V$ volume), with infinite volume. The purpose of this hand-wavy argument is simply to point to the fact that more constrained classes of operations may include further physical restrictions on the properties of the bath that can be accessed, and that within Thermal Operations we make no restriction on the environment beyond its thermality. For example, Gibbs-stochastic matrices $G$ acting only on two energy levels at a time can be realised as $G^\mathcal{T}$ even if the bath is restricted to be a single-mode bosonic thermal state (Appendix B of Ref.~\cite{lostaglio2018elementary}). Studies of finite baths from a resource theory perspective have also been conducted \cite{reeb2014improved, richens2017finite, sparaciari2017resource}.
\end{rmk}
We can move on to use the previous theorem to derive the laws of thermodynamics governing the changes of populations. We will see that, rather elegantly, the result has a strict connection with Nielsen's theorem in the theory of entanglement. 

\subsection{Ordering states: from entropy to majorisation}

\subsubsection{The central theorem on majorisation}

Recall the technical problem we need to solve: given probability distributions $\v{p}$ and $\v{p'}$ (same dimension) give necessary and sufficient conditions for the existence of a Gibbs-stochastic matrix $G$ such that $G \v{p} = \v{p}'$. In order to solve this problem we will take a detour and first solve the infinite temperature limit $\beta \rightarrow 0$, or equivalently the trivial Hamiltonian case $H_S \propto \I$, in which  the fixed point is $\v{g} = \v{1}/n = (1/n,...,1/n)$ \cite{gour2015resource}. Interestingly, we will see that in this limit the partial order that emerges is the same (precisely, the  opposite) compared to that defined on the set of bipartite pure entangled states by Local Operations and Classical Communication (LOCC, see Example~\ref{ex:locc}).

In this section we will then look at the problem of the existence of a stochastic matrix $B$ such that $B \v{p} = \v{p}'$, \mbox{$B \v{1}/n = \v{1}/n$}, for any two probability vectors $\v{p}$, $\v{p}'$. Stochastic matrices satisfying $B \v{1} = \v{1}$, or equivalently non-negative matrices $B$ satisfying $\sum_{i} B_{j|i} = \sum_j B_{j|i}=1$, are known as \emph{doubly-stochastic} or \emph{bistochastic} matrices.  

In thermodynamics we are used to define a `state function', the entropy, to measure the amount of `disorder' in the system. In Jaynes' terms, once we define what variables are under control, we introduce a function measuring our lack of knowledge within the given description \cite{jaynes1957information}. While in equilibrium considerations this function is taken to be the Gibbs/Shannon entropy \cite{jaynes1965gibbs}, we will need a more refined version of that concept. In the setting under consideration, this is captured by the notion \emph{majorisation}. Given a vector $\v{x} \in \R^n$,  denote by $\v{x}^{\downarrow}$ the vector $\v{x}$ sorted in non-increasing order. Then

\begin{defn}[Majorisation]
	\label{def:majorisation}
	$\v{x}$ majorises $\v{y}$, denoted $\v{x} \succ \v{y}$, if and only if
	\be
	\sum_{i=1}^k x^\downarrow_i \geq \sum_{i=1}^k y^{\downarrow}_i, \quad k=1,...,n-1, \quad 	\sum_{i=1}^n x_i = \sum_{i=1}^n y_i.
	\ee
\end{defn}
The relation $\succ$ defines a partial ordering among vectors in $\R^n$.
\begin{ex}
	\label{ex:majorisationfirstexample}
	Let $\v{1}/3 = (1/3,1/3,1/3)$ and $\v{e}_1 = (1,0,0)$. Then for all  probability distributions $\v{x} \in \R^3$, we have $\v{e}_1 \succ \v{x} \succ \v{1}/3$. However, $\succ$ is not a total ordering. Take $\v{y} = (2/3, 1/6, 1/6)$, $\v{z} = (1/2,1/2,0)$. Then, neither $\v{y} \succ \v{z}$ nor $\v{z} \succ \v{y}$. 
\end{ex}
An alternative, more geometrical, definition is easily seen to be equivalent to the previous one:
\begin{defn}[Lorenz curves]
	\label{defn:majorisationcurve}
	Let $L(\v{x})$ be the piecewise linear curve in $\R^2$ obtained by joining the points $\left(k, \sum_{i=1}^k x^{\downarrow}_i\right)$, for $k=1,...,n$. We say that $L(\v{x}) \succ L(\v{y})$ if and only if the curve $L(\v{x})$ all lies not below $L(\v{y})$ and the two curves end at the same height.
\end{defn}
The last requirement becomes trivial if $\v{x}$ and $\v{y}$ are both probability distributions. The curve $L(\v{x})$ is called \emph{Lorenz curve} of $\v{x}$ and it is not difficult to check that $L(\v{x}) \succ L(\v{y})$ if and only if $\v{x} \succ \v{y}$ (see Fig.~\ref{fig:majorisation_ex}).   

	\begin{figure}[h!]
	\centering
	\includegraphics[width=0.4\columnwidth]{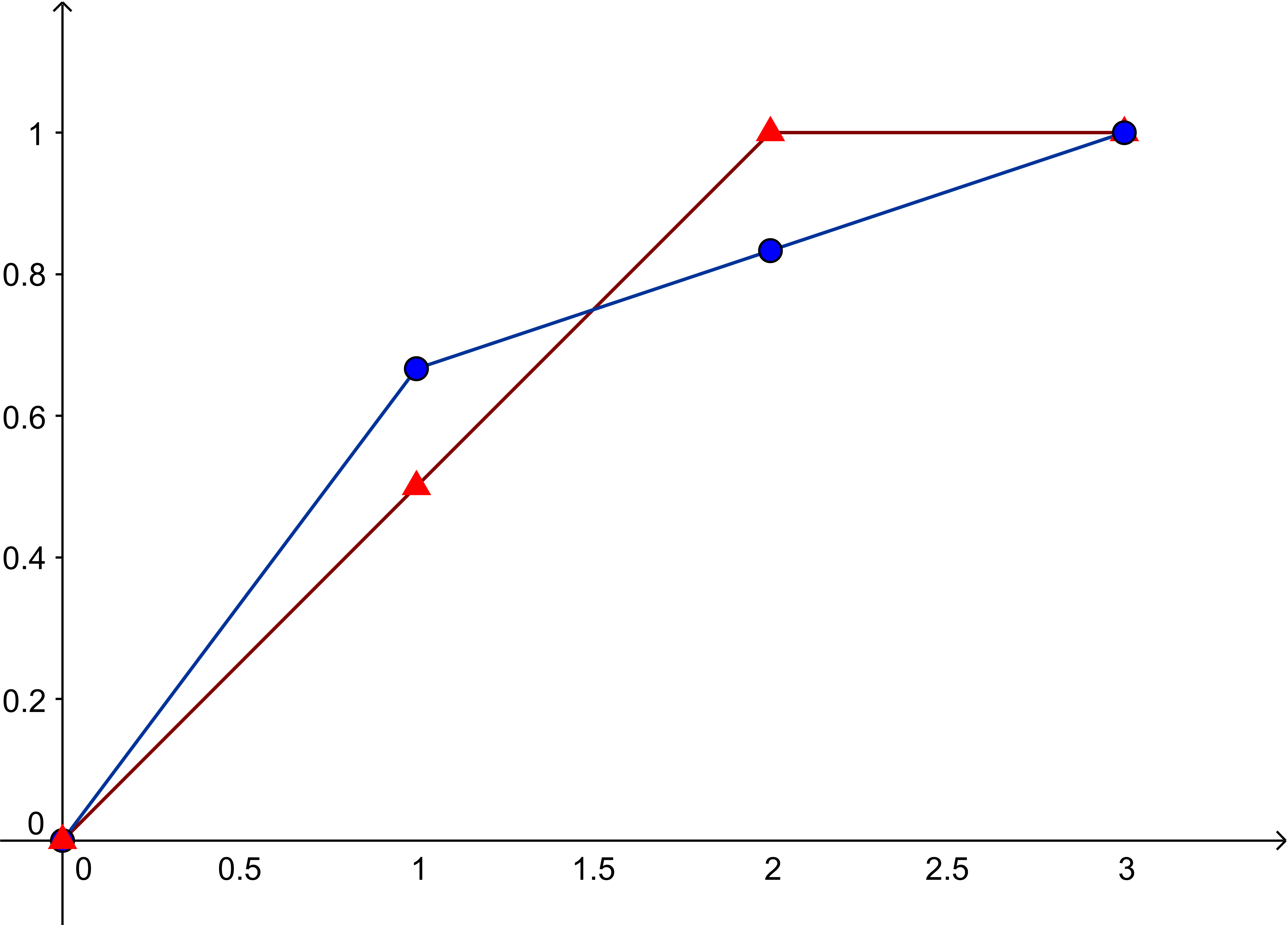}
	\caption[Majorisation curves.]{\label{fig:majorisation_ex} {\bf Majorisation (Lorenz) curves.} Given any probability distribution $\v{x}$, we can define a majorisation curve as in Def.~\ref{defn:majorisationcurve}. A distribution $\v{x}$ majorises $\v{y}$ if and only if the Lorenz curve of $\v{x}$ is all above that of $\v{y}$. Neither of the two distributions of the picture majorises the other, since their corresponding Lorenz curves intersect.} 
\end{figure}
Majorisation is the partial ordering defined by LOCC on the spectra of the reduced states of pure bipartite states:
\begin{thm}[Nielsen \cite{nielsen2010quantum}]
	Let $\ket{\psi}_{AB}$, $\ket{\phi}_{AB}$ be pure, bipartite quantum states. Then there exists a LOCC transformation mapping  $\ket{\psi}_{AB}$ into $\ket{\phi}_{AB}$ if and only if $\psi \prec \phi$, where $\psi$ ($\phi$) is the spectrum of the reduced state, over either $A$ or $B$, of $\ket{\psi}_{AB}$ ($\ket{\phi}_{AB}$).
\end{thm}

It turns out that majorisation also gives the solution to the infinite temperature problem introduced above~\cite{marshall2010inequalities, bhatia2013matrix}:

\begin{thm}[Hardy, Littlewood, Polya]
	\label{thm:majorisationbistochastic}
	$\v{x} \succ \v{y}$ if and only if there exists a stochastic matrix $B$ satisfying 
	\be
	\label{eq:Mbistochastic}
	B\v{x} = \v{y}, \quad B\v{1} = \v{1}.
	\ee 
\end{thm}
  	\begin{proof} Since we will need to use this result, to maintain these notes self-contained we give the proof in the Appendix (see also \cite{marshall2010inequalities}, Chapter~2).
 	 	\end{proof}

Note that the partial ordering defined by the notion of accessibility through doubly-stochastic matrices is exactly the same (but inverted) with respect to pure bipartite entangled states under LOCC.

\subsubsection{Entropies vs majorisation}

 The fact that the uniform distribution $\v{1}/n$ is a fixed point can be intuitively understood noticing that this is the state of maximum Shannon entropy, and a process $M$ with $M \v{1} \neq \v{1}$ would decrease entropy. However, it is fruitful to make the connection between majorisation and the concept of entropy more precise. To do so, we define a natural class of functions -- those that preserve the partial order structure of majorisation:
 \begin{defn}
 	A function $f: \R^n \rightarrow \R$ is called Schur-convex if and only if
 	\be
 	\v{x} \succ \v{y} \Rightarrow f(\v{x}) \geq f(\v{y}),
 	\ee
 	and Schur-concave if and only if $	\v{x} \succ \v{y} \Rightarrow f(\v{x}) \leq f(\v{y}))$.
 \end{defn}
 $f$ is a homomorphism  from the partially ordered set $(\R^n,\succ)$ to the totally ordered set of real numbers. In fact, we can think of each Schur-concave $f$ as a possible entropy functional, since any physical process with the uniform distribution as a fixed point will not decrease it. Unsurprisingly, each $f$ can only capture some aspects of the partial order. In the context of resource theories, $f$ is called a \emph{monotone} under the set of allowed operations (in this case, doubly-stochastic maps). Note that, in terms of capturing the partial order induced by majorisation, two Schur-concave (or convex) functions $f$ and $\tilde{f}$ such that $\tilde{f} = g \circ f$, with $g$ strictly non increasing function,  ought to be regarded as equivalent, since $f(\v{x}) \geq f(\v{y}) \Leftrightarrow \tilde{f}(\v{x}) \geq \tilde{f}(\v{y})$ for all $\v{x}$ and $\v{y}$. 
 
  \begin{rmk} [Schur-concave functions toolbox] 	\label{rmk:schurconcavetoolbox} Schur-concave functions on $(\R^n, \succ)$ can be constructed from concave functions on $\R$. Let $h: \R \rightarrow \R$ be concave (convex). Then using Theorem~\ref{thm:majorisationbistochastic} and the concavity of the function $f$, one can show that $f(\v{x}) = \sum_{i=1}^n h(x_i)$ is Schur-concave (Schur-convex). Examples of such functions include the Shannon entropy and all the R\'enyi entropies \cite{renyi1961measures}.
\end{rmk}

 The reason why majorisation is a more refined concept than entropy may be understood through an example. Consider the Shannon entropy, \mbox{$H(\v{x}) = - \sum_{i=1}^n x_i \log x_i$}. Since $-x\log x$ is concave, $H(\v{x})$ is Schur-concave (see Remark~\ref{rmk:schurconcavetoolbox}). So, $H(\v{x}) \leq H(\v{y})$ is a necessary condition for $\v{x} \succ \v{y}$. However, it is \emph{not} sufficient, as the following example shows:
 
 \begin{ex}
 	Consider $\v{y}$ and $\v{z}$ from Example~\ref{ex:majorisationfirstexample}. From Theorem~\ref{thm:majorisationbistochastic}, no stochastic process with the uniform distribution as a fixed point can map $\v{z}$ into $\v{y}$, even though \mbox{$H(\v{z})<H(\v{y})$}. 
 \end{ex}
 
 When coupled to Theorem~\ref{thm:majorisationbistochastic}, this tells us something important. The decrease of $H(\v{x})$, while necessary, does \emph{not} guarantee the existence of a `mixing process' represented by a doubly-stochastic map. Doubly stochastic maps can be regarded as mixing processes because, due to Birkhoff's theorem, [\cite{marshall2010inequalities}, Chapter 2], they coincide with the convex mixtures of permutations. $\v{x} \succ \v{y}$ gives stronger constraints than $H(\v{x}) \leq H(\v{y})$; in fact, it gives exactly the constrains that guarantee the existence of a mixing process mapping $\v{x}$ into $\v{y}$. There are hence many `entropies': a prominent example are the abovementioned R\'enyi entropies \cite{renyi1961measures}
\begin{equation}
\label{eq:renyientropies}
H_\alpha(\v{x}) = \frac{\rm sgn(\alpha)}{1-\alpha} \log \sum_i x^\alpha_i, 
\end{equation}
that can be proven to be Schur-concave using Remark~\ref{rmk:schurconcavetoolbox} and elementary properties of the functions $x \mapsto x^\alpha$. However, not all entropies are on the same foot. In fact, the Shannon entropy $H(\v{x})$ is distinguished as the unique monotone when we discuss the interconversion between a large number of copies of the initial state to a large number of copies of the final state:
\begin{equation}
\label{eq:asymptoticrate}
\v{x}^{\otimes N} \rightarrow \v{y}^{\otimes M}.
\end{equation}
When an $\epsilon$-error in $\ell_1$ norm is allowed, from typicality arguments the maximum achievable ratio $M/N$ equals $\frac{\log n - H(\v{x})}{\log n - H(\v{y})}$, where $n$ is the dimension of $\v{x}$ and $\v{y}$. In a precise sense, all entropies `converge' to the Shannon entropy, see Lemmas 65-67 of Ref.~\cite{gour2015resource}.
\begin{rmk}[Intuition on asymptotic conversion rates]
	\label{rmk:asymptoticinterconvesionTinfinity}
	To get some intuition on asymptotic interconversion rates, take $\v{x} = (x,1-x)$ (call the two states $0$ and $1$) and write down the distribution $\v{x}^{\otimes N}$. When $N\rightarrow +\infty$ the distribution over the number of zeros in the string becomes a Gaussian arbitrarily sharply peaked around $x N$ zeros. Consider only the strings with a number of zero equal to $x N$. These are $\binom{N}{xN}$; using Stirling's approximation one can show these are $ \approx e^{H(\v{x})N}$, so that $\v{x}^{\otimes N}$ is well approximated by a uniform distribution over this number of typical strings. Similarly one can reason on $\v{y}$. Adding zeros to make the two distributions of equal dimension and applying the majorisation condition one finds the largest ratio $M/N$ such that a transition between these two approximate versions of $\v{x}$ and $\v{y}$ is possible.
	\end{rmk}

\subsection{Ordering non-equilibrium: from free energy to thermo-majorisation} 

\subsubsection{A Nielsen's theorem for thermal stochastic processes}

We have seen that majorisaton is equivalent to the existence of stochastic processes having the uniform distribution as a fixed point of the dynamics (see Theorem~\ref{thm:majorisationbistochastic}). As discussed, these results can be understood, from a thermodynamic perspective, as being valid when the temperature $T$ of the environment is $T = +\infty$ or the Hamiltonian $H_S$ is trivial. Luckily, these technical results can be extended to any finite temperature and non-trivial Hamiltonian. Conceptually this will lead to a generalised notion of free energy, in the same way in which the previous considerations led us to a generalised notion of entropy.

The basic tool will be the embedding map introduced in Ref.~\cite{brandao2013second}, which loosely speaking one can understand as connecting the microcanonical and macrocanonical ensembles (see Ref.~\cite{egloff2015measure}, Appendix A). Assume the thermal state $\v{g}$ is a vector of rational numbers, i.e. there exists $d_1,...,d_n \in \mathbb{N}$:
\be
\label{eq:rationalgibbs}
\v{g} = \left(\frac{d_1}{D}, ...,\frac{d_n}{D}\right),
\ee
where $D:= \sum_{i=1}^n d_i$. Of course, any irrational $\v{g}$ can be approximated to an arbitrary precision as in Eq.~\eqref{eq:rationalgibbs} (we will ignore here some technicalities and assume that $\v{g}$ is rational). Then, if $\v{d} := (d_1,...,d_n)$, we define
\begin{defn}[Embedding map]
	\label{defn:embeddingmap}
	$\Gamma_{\v{d}}:\R^n \rightarrow \R^D$ is the function
	\be
	\label{eq:embeddingmap}
	\Gamma_{\v{d}}(\v{x}) := \left(\underbrace{\frac{x_1}{d_1},...,\frac{x_1}{d_1}}_{d_1-\textrm{times}},...,\underbrace{\frac{x_n}{d_n},...,\frac{x_n}{d_n}}_{d_n-\textrm{times}}\right):=\oplus_i x_i \v{1}_i/d_i,
	\ee
	with $\v{1}_i/d_i$ a $d_i$-dimensional uniform distribution.
\end{defn}
By definition, from Eq.~\eqref{eq:rationalgibbs} it follows $\Gamma_{\v{d}}(\v{g}) = \v{1}/D$, where $\v{1}/D$ is the $D-$dimensional uniform distribution. The basic idea is to map dynamics with a thermal fixed point to dynamics with a uniform fixed point in a larger space. The (left) inverse of $\Gamma_{\v{d}}$ is the map  $\Gamma^{-1}_{\v{d}}:\R^D \rightarrow \R^n$ defined by
\be
\Gamma^{-1}_{\v{d}}(\v{p})= \v{x},
\ee
where $x_i = \sum^{j_{i}}_{j=j_{i-1}+1} p_j$, $j_i = \sum_{k=0}^{i} d_k$ and $d_0 := 0$, for $i=1,...,n$. This simply amounts to taking the various blocks on the right-hand side of Eq.~\eqref{eq:embeddingmap} and summing over the elements within each block. Then, $ \Gamma^{-1}_{\v{d}}(\Gamma_{\v{d}}(\v{x})) = \v{x}$ for all $\v{x} \in \R^n$ (but, conversely, $\Gamma_{\v{d}}(\Gamma^{-1}_{\v{d}}(\v{x}))$ is not the identity on $\R^D$). 

The embedding is a bridge to majorisation, as the following lemma shows (see Fig.~\ref{fig:bridgelemma}):
\begin{lem}
	\label{lem:embeddingmajorisation}
	There exists a Gibbs-stochastic map $G$ such that $G\v{x} = \v{y}$ if and only if $\Gamma_{\v{d}}(\v{x}) \succ \Gamma_{\v{d}}(\v{y})$. 
\end{lem}
\begin{proof}
	One can verify that a Gibbs-stochastic mapping $G$ from $\v{x}$ to $\v{y}$ exists if and only if there exists a doubly-stochastic map $B$ transforming $\Gamma_{\v{d}}(\v{x})$ into $\Gamma_{\v{d}}(\v{y})$. To see this, simply  define $B = \Gamma_{\v{d}} \circ G \circ \Gamma^{-1}_{\v{d}}$ if $G$ is given or $G =  \Gamma^{-1}_{\v{d}} \circ B \circ \Gamma_{\v{d}}$ if $B$ is given (note that the composition of stochastic matrices is stochastic). Using Theorem~\ref{thm:majorisationbistochastic}, the result follows.
\end{proof}

\begin{figure}[h!]
	\centering
	\includegraphics[width=0.8\columnwidth]{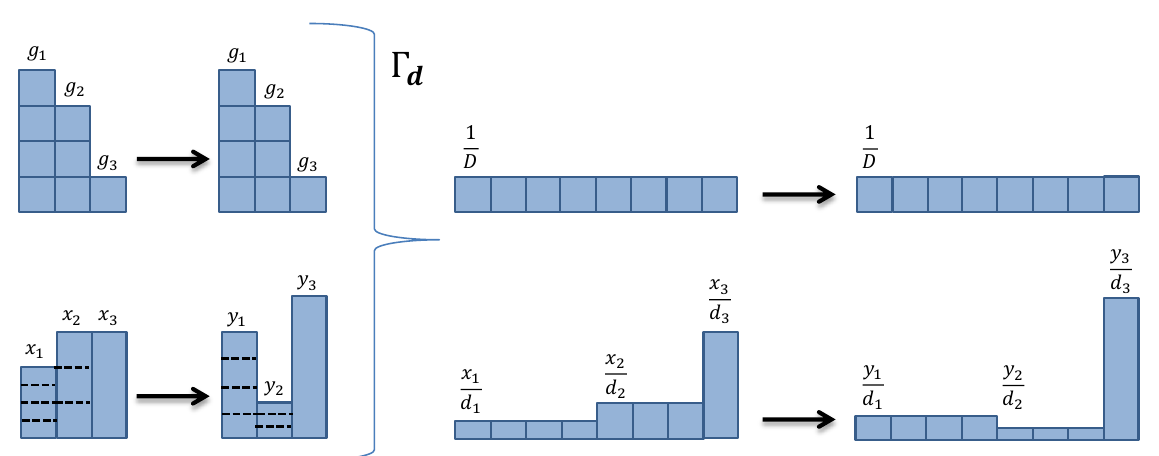}
	\caption{\label{fig:bridgelemma} {\bf The bridge lemma.} The question of the existence of a stochastic map transforming the vector $\v{x}=(x_1,x_2,x_3)$ into the vector $\v{y}=(y_1,y_2,y_3)$ while leaving a vector $\v{g}=(g_1,g_2,g_3)$ unchanged is equivalent to the question of the existence of a stochastic map transforming the vector $\Gamma_{\v{d}}(\v{x}) =(x_1/d_1,...,x_1/d_1,x_2/d_2,...,x_2/d_2,x_3/d_3,...,x_3/d_3)$ into $\Gamma_{\v{d}}(\v{y}) = (y_1/d_1,...,y_1/d_1,y_2/d_2,...,y_2/d_2,y_3/d_3,...,y_3/d_3)$ while leaving the uniform vector $(1/D,...,1/D)$ unchanged. Here $\v{d} = (d_1,d_2,d_3) = (4,3,1)$, $D = 8$, since $g_1 = 1/2$, $g_2 = 3/8$, $g_3 = 1/8$. The picture shows with dashed line the division of each $x_i$ into $d_i$ parts, and the construction of the embedded distribution $x_i/d_i$ from them.} 
\end{figure}

From this lemma we can define a relation that generalises the notion of majorisation. In particular, we will see how the condition that the embedded distribution $\Gamma_{\v{d}}(\v{x})$ majorises $\Gamma_{\v{d}}(\v{y})$ can be rephrased as a \emph{thermo-majorisation} condition involving only $\v{x}$ and $\v{y}$.

Let $x^{\downarrow \beta}_i$ be the so-called $\beta-$ordering of $\v{x}$, defined as the rearrangement of the indices $i$ such that the vector $x_i/g_i$ is sorted in non-increasing order. In other words, $x^{\downarrow \beta}_i = x_{\pi(i)}$, where $\pi$ is the permutation ensuring $x_{\pi(1)}/g_{\pi(1)} \geq x_{\pi(2)}/g_{\pi(2)} \geq \dots \geq x_{\pi(n)}/g_{\pi(n)}$. 
\begin{ex}
	\label{ex:betaorder}
	Consider the Hamiltonian with spectrum $E^S_1 = 0$, $E^S_2 = 1$, $E^S_3 =2$ and let $\beta = 1.2$. Hence, \mbox{$\v{g} = (0.718436, 0.216389, 0.0651751)$}. Let \mbox{$\v{x} = (1/3,1/3,1/3)$} and \mbox{$\v{y} = (2/3,1/3,0)$}. Then $\v{x}^{\downarrow \beta} = (x_3,x_2,x_1)$ and $\v{y}^{\downarrow \beta} = (y_2,y_1,y_3)$.
\end{ex}
Then define \cite{horodecki2013fundamental}
\begin{defn}[Thermo-majorisation curves]
	\label{defn:thermomajorisation}
	Let $T(\v{x})$ be the piecewise linear curve in $\R^2$ obtained by joining the origin and the points $\left(Z_S \sum_{i=1}^k g_{\pi(i)}, \sum_{i=1}^k x_{\pi(i)}\right)$, for $k=1,...,n$ (where $\pi$ is the permutation that $\beta$-orders $\v{x}$). $T(\v{x})$ is called thermo-majorisation curve of $\v{x}$. We say that $\v{x}$ thermo-majorises $\v{y}$, denoted $\v{x} \succ_g \v{y}$, if and only if the curve $T(\v{x})$ all lies not below $T(\v{y})$ and the two curves end at the same height.
\end{defn}

As with Lorenz curves, the last requirement is trivial if $\v{x}$ and $\v{y}$ are probability distributions. Note that thermo-majorisation curves are the same as Lorenz curves when $\v{g}$ is uniform (in fact, in the mathematics literature thermo-majorisation is known as majorisation relative to $\v{g}$, or $d$-majorisation \cite{marshall2010inequalities}). Are there a `top' and `bottom' states in the thermo-majorisation ordering? One can verify that the state $(0,\dots,0,1)$ (sharp state with largest energy) is the top, and the thermal state $\v{g}$ is the bottom, in the sense that $(0,\dots,0,1) \succ_g \v{x} \succ_g \v{g}$ for all $n$-dimensional probability distributions $\v{x}$.

If $\v{d}$ is the vector related to $\v{g}$ by Eq.~\eqref{eq:rationalgibbs}, the following lemma holds:
\begin{lem}
	\label{lem:bridge}
	$\Gamma_{\v{d}}(\v{x}) \succ \Gamma_{\v{d}}(\v{y})$ if and only if $\v{x} \succ_g \v{y}$.
\end{lem}
\begin{proof}
	Sorting in decreasing order the $D$-dimensional probability distributions $\Gamma_{\v{d}}(\v{x})$ and $\Gamma_{\v{d}}(\v{y})$ corresponds to $\beta-$ordering the $n$-dimensional probability distributions $\v{x}$ and $\v{y}$. Then we can use that $\Gamma_{\v{d}}(\v{x}) \succ \Gamma_{\v{d}}(\v{y})$ if and only if $L(\Gamma_{\v{d}}(\v{x})) \succ L(\Gamma_{\v{d}}(\v{y}))$. If we remove from the points used to construct $L(\Gamma_{\v{d}}(\v{x}))$ all the non-extremal points that lie on a segment of given slope (the ``non-elbow'' points), we obtain the same Lorenz curve, see Fig.~\ref{fig:elbows}. In particular, instead of joining all points $\left(k, \sum_{i=1}^k (x_i/d_i)^{\downarrow}\right)$, \mbox{$k=1,..D$}, we can just join the points at the ``elbows''; i.e., if $\pi$ is the permutation that $\beta$-orders $\v{x}$, define $k_s = \sum_{i=1}^s d_{\pi(i)}$ and join $\left(k_s, \sum_{i=1}^{k_s} \Gamma^{\downarrow}_{\v{d}}(\v{x})_i\right)$, $s=1,...,n$ (as well as the origin). But $\left(k_s, \sum_{i=1}^{k_s} \Gamma^{\downarrow}_{\v{d}}(\v{x})_i\right) = \left(D\sum_{i=1}^s g_{\pi(i)}, \sum_{i=1}^s x_{\pi(i)} \right)$, which is the same as $T(\v{x})$ apart for a rescaling of the $x$-axis. Repeat the same reasoning for $L(\Gamma_{\v{d}}(\v{y}))$ and notice that the rescaling of the $x$-axis is the same for all curves (and hence it does not affect comparisons). Hence we conclude $L(\Gamma_{\v{d}}(\v{x})) \succ L(\Gamma_{\v{d}}(\v{y}))$ if and only if $T(\v{x})) \succ T(\v{y})$, i.e.  $\v{x} \succ_g \v{y}$, which concludes the proof.
\end{proof}

\begin{figure}[h!]
	\centering
	\includegraphics[width=0.5\columnwidth]{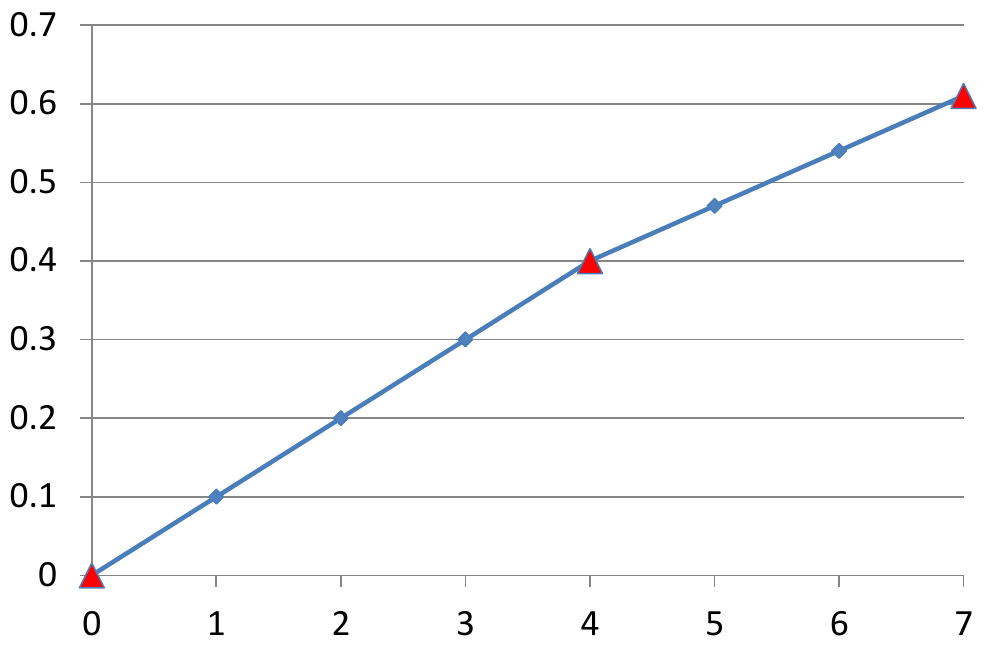}
	\caption{\label{fig:elbows} {\bf Removing non-elbow points.} Given the majorisation curve of $\Gamma_{\v{d}}(\v{x}) $, we can distinguish between elbows points (red triangles) and non-elbow points (blue dots). In this example we took $x_{\pi(1)}/d_{\pi(1)} =0.1$, $d_{\pi(1)} = 4$, $x_{\pi(2)}/d_{\pi(2)} =0.06$  and $d_{\pi(2)} = 3$ (only part of the majorisation curve is presented). One obtains exactly the same curve by connecting only the elbow points.} 
\end{figure}

Putting together all these results
\begin{thm}
	$\v{x} \succ_g \v{y}$ if and only if there exists a stochastic matrix $G$ satisfying
	\begin{equation}
	G \v{x} = \v{y}, \quad G\v{g} = \v{g}.
	\end{equation}
\end{thm}
An equivalent form is given in Theorem~1 of Ref.~\cite{ruch1980generalisation} (the equivalence follows, e.g., from Lemma 2.15 of Ref.~\cite{lostagliothesis}).  As a corollary of the above and Theorem~\ref{th:thermalgibbs}, together with the fact that dephasing in the energy basis, Eq.~\eqref{eq:dephasing}, is a Thermal Operation, we get the following thermal Nielsen's theorem:
\begin{thm}[Thermal Nielsen's theorem]
\label{thm:thermalnielsen} 
If $\v{x}$ is the population of $\rho_S$, there exist Thermal Operations $\mathcal{T}_\epsilon$ such that $\mathcal{T}_\epsilon(\rho_S)$ has population arbitrarily close to $\v{y}$ if and only if $\v{x} \succ_g \v{y}$. 
\end{thm}

\begin{rmk}
	\label{rmk:simplextension}
	This result fully solves the interconversion problem whenever $[\rho_S, H_S] = 0$. In fact, if $[\sigma_S, H_S] = 0$, then to verify if for every $\epsilon> 0$ there exists $\mathcal{T}_\epsilon$ with $\mathcal{T}_\epsilon(\rho_S) \approx_\epsilon \sigma_S$ we simply check if the population of $\rho_S$ thermo-majorises the population of $\sigma_S$; and if $[\sigma_S, H_S] \neq 0$ we know that no Thermal Operation exists, because the commutation of Thermal Operations with the dephasing map gives $[\mathcal{T}(\rho_S), H_S] = 0$ for every Thermal Operation $\mathcal{T}$. We will come back to the question of adding superpositions into the picture in Section~\ref{sec:coherence}.
\end{rmk}


\subsubsection{Free energy vs thermo-majorisation, `second laws' and catalysis}
\label{sec:thermodynamicschur}

In the same way in which we defined Schur-concave and convex functions as those preserving the majorisation ordering, we can define functions that preserve the thermo-majorisation ordering. In the absence of a generally agreed name for such functions, we call them thermodynamic Schur-concave functions (or $\v{g}$-Schur-concave functions for short):
\begin{defn}
	A function $f: \R^n \rightarrow \R$ is called $\v{g}$-Schur-convex (respectively, $\v{g}$-Schur-concave) if and only if
	\be
	\v{x} \succ_g \v{y} \Rightarrow f(\v{x}) \geq f(\v{y}) \quad (\textrm{respectively,} \; \; f(\v{x}) \leq f(\v{y})).
	\ee
\end{defn}
If Schur-concave functions are akin to entropies, thermodynamic Schur-convex functions are akin to free energies, each capturing some aspect of the ordering.

\begin{rmk}[$\v{g}$-Schur-convex functions toolbox]
	\label{rmk:gschurtoolbox}
	As before, we can give a tool to construct $\v{g}$-Schur-concave functions on $(\R^n, \succ)$ from concave functions on $\R$: Let $h: \R \rightarrow \R$ be concave (convex). Then the function $f$
	\be
	f(\v{x}) = \sum_{i=1}^n g_i h\left(\frac{x_i}{g_i}\right),
	\ee
	also known as $f$-divergence, is \v{g}-Schur-concave (\v{g}- Schur-convex). 
		We prove the statement for $h$ convex (the other case is the same). As we have seen, $\v{x} \succ_g \v{y}$ if and only if $y_i = \sum_{j=1}^n G_{i|j} x_j$, with \mbox{$\sum_{j=1}^n G_{i|j} \frac{g_j}{g_i} = 1$} and \mbox{$\sum_{i=1}^n G_{i|j} =1$}. Then,
		{\small
			\be
			\nonumber
			f(\v{y}) = \sum_{i=1}^n g_i h\left(\sum_{j=1}^n G_{i|j} \frac{x_j}{g_i}\right) =  \sum_{i=1}^n g_i h\left(\sum_{j=1}^n \left[G_{i|j}\frac{g_j}{g_i}\right] \frac{x_j}{g_j}\right) \leq \sum_{j=1}^n g_j h\left(\frac{x_j}{g_j}\right) = f(\v{x}).
			\ee
		}
\end{rmk}

Following the same discussion given for majorisation, one can argue that thermo-majorisation is a more refined concept than the standard constraint of decreasing the (non-equilibrium) free energy. This can be seen as follows. Define \mbox{$F(\v{x}) = U(\v{x}) - k T H(\v{x})$}, where \mbox{$U(\v{x}) = \sum_i x_i E^S_i$} is the average energy. Because $x \log x$ is convex, $F(\v{x})$ is $\v{g}$-Schur-convex. Hence, if there exists a Gibbs-stochastic map transforming $\v{x}$ into $\v{y}$, i.e., if $\v{x} \succ_g \v{y}$, we must have $F(\v{x}) \geq F(\v{y})$. However, the decrease of the free energy $F$ does \emph{not} guarantee the existence of such physical process, as the following example shows:

\begin{ex}
	\label{ex:freenergydownbut}
	Consider the Hamiltonian and states of Example~\ref{ex:betaorder}. Then $F(\v{x}) \approx 0.084 >F(\v{y}) \approx -0.197$. Nevertheless, $T(\v{x})$ crosses $T(\v{y})$, so there is no map transforming $\v{x}$ into $\v{y}$ while leaving $\v{g}$ fixed. To see this, recall that $\v{x}^{\downarrow \beta} = (x_3,x_2,x_1)$ and $\v{y}^{\downarrow \beta} = (y_2,y_1,y_3)$. So, to obtain $T(\v{x})$ we need to join $\{(0,0),(e^{-2.4},1/3), (e^{-2.4 } + e^{-1.2},2/3), \\ (e^{-2.4} + e^{-1.2} +1,1)\}$; and to get $T(\v{y})$ we need to join $\{(0,0),(e^{-1.2},1/3),(e^{-1.2}+1,1),(e^{-2.4}+ e^{-1.2}+1,1)\}$ (see Fig.~\ref{fig:thermomajorisation_ex}).
	
\end{ex}

\begin{figure}[h!]
	\centering
	\includegraphics[width=0.5\columnwidth]{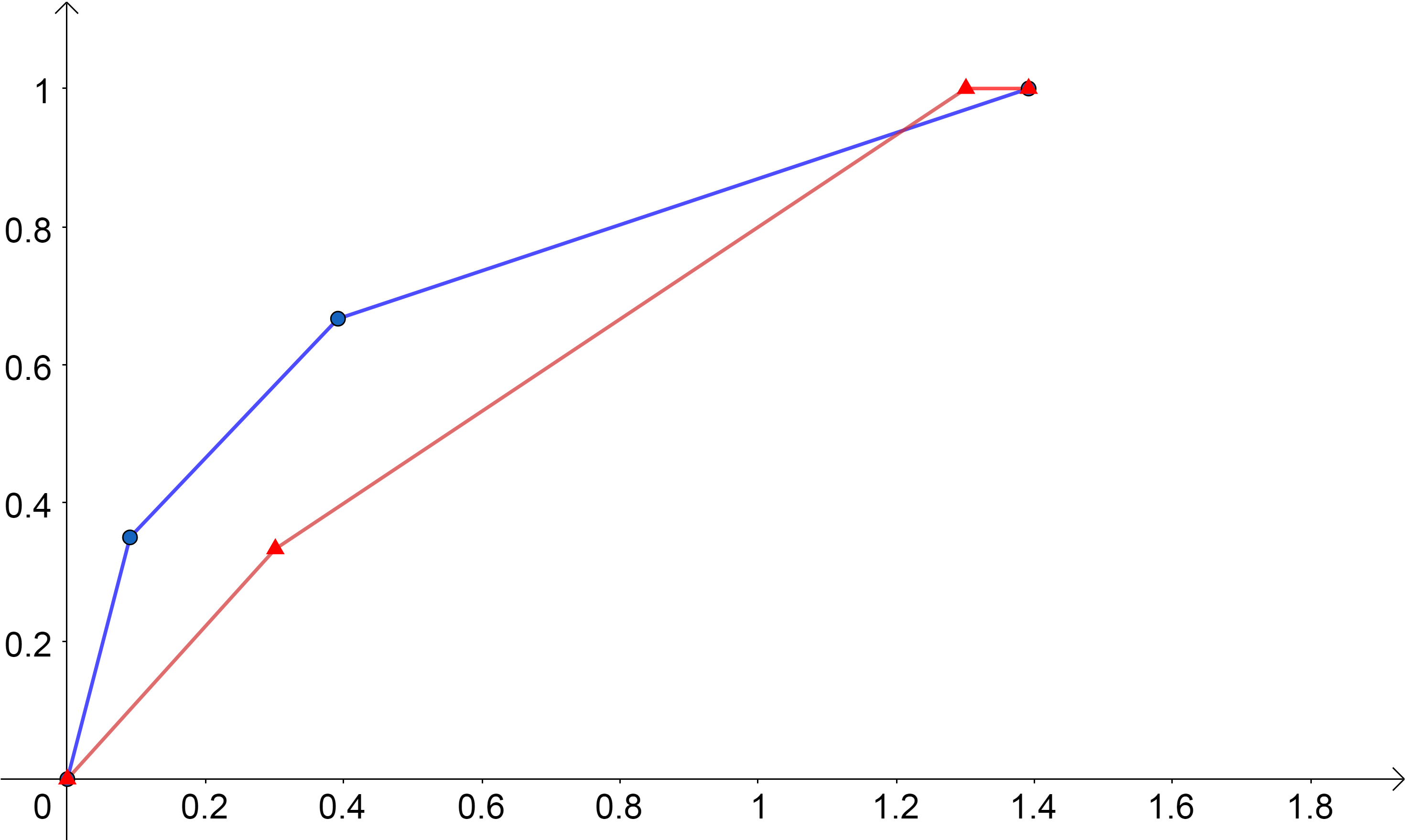}
	\caption[Thermo-majorisation curves vs free energy.]{\label{fig:thermomajorisation_ex} {\bf Thermo-majorisation curves vs free energy.} The thermo-majorisation curves of $\v{x}$ and $\v{y}$ from Example~\ref{ex:freenergydownbut}, denoted by $T(\v{x})$ and $T(\v{y})$, respectively. $T(\v{x})$ is the curve connecting the blue dots, whereas $T(\v{y})$ is obtained connecting the red triangles. Despite $F(\v{x})> F(\v{y})$, $T(\v{x})$ does not lie all above $T(\v{y})$. Hence, there is no stochastic process $G$ with $G\v{g} = \v{g}$ and $G\v{x} = \v{y}$ ($\v{g}$ here is the Gibbs distribution).} 
\end{figure}

\begin{rmk}[$\alpha$-free energies]
		The \mbox{$\alpha$-free} energy of $\v{x}$ is defined as
		\be
		\label{eq:alphafreeenergy}
		F_{\alpha}(\v{x}) = - kT \log Z_S + k T S_{\alpha}(\v{x}\|\v{g}), \quad S_{\alpha}(\v{x}\|\v{y}) =  \frac{{\rm sgn}(\alpha)}{\alpha-1}\log \sum_i x^{\alpha}_i y^{1-\alpha}_i,
		\ee
		where $S_{\alpha}$ are the so-called $\alpha$-R{\'e}nyi divergences \cite{renyi1961measures}.
		The cases $\alpha\in\{-\infty,0,1,+\infty\}$ are defined via suitable limits (see e.g.~\cite{brandao2013second}):
		\be
		\nonumber
		S_{\infty}(\v{x}\|\v{g}) = \log \max_i x_i/g_i, \quad S_1(\v{x}\|\v{g}) = \sum_i x_i \log (x_i/g_i), 
		\ee
		\be
		\label{eq:zerorenyi}
		S_0(\v{x}\|\v{g}) = - \log \sum_{i | x_i \neq 0} g_i, \quad S_{-\infty}(\v{x}\|\v{g}) = S_\infty(\v{g}\|\v{x}).
		\ee
	In particular notice that $F_1(\v{x}) = F(\v{x})$, as defined above, and $F_\alpha(\v{g}) = -kT \log Z_S$ for all $\alpha$. All $F_{\alpha}$ must monotonically decrease under Gibbs-stochastic maps (as one can derive using Remark~\ref{rmk:gschurtoolbox}). Also note that a direct calculation shows that $F_\alpha$ are related to the R\'enyi entropies $H_\alpha$ by the embedding map introduced above \cite{brandao2013second}:
	\begin{equation}
	\label{eq:relation}
F_{\alpha}(\v{x}) + kT \log Z_S = kT(\log D - H_\alpha(\Gamma_{\v{d}}(\v{x}))).
	\end{equation} 
\end{rmk}
From the thermal Nielsen's theorem we get that if there is a Thermal Operation mapping $\rho_S$ into $\sigma_S$, with populations $\v{x}$ and $\v{y}$, then we necessarily have
\begin{equation}
\label{eq:secondlaws}
F_{\alpha}(\v{x}) \geq F_{\alpha}(\v{y}), \quad \forall \alpha \in \R.
\end{equation}
Note that these are not sufficient (thermo-majorisation imposes stricter conditions). The decrease of all $\{F_\alpha\}$, together with $F_{\rm Burg}(\v{x}) = kT S(\v{g}\| \v{x}) - kT \log Z_S$, becomes sufficient to the existence of a physical map between diagonal states only when \emph{catalysts} are allowed, i.e. when one can introduce states that aid the transformation without being degraded in the process. This was the main result of Ref.~\cite{brandao2013second}: 
\begin{thm}[ ``Second laws'' are sufficient in the presence of catalysts]
	Suppose $\v{x}$ and $\v{y}$ have full support and $\v{x} \neq \v{y}$. Then there exists a state $\v{c}$ with 
	\begin{equation}
	\v{x} \otimes \v{c} \succ_g \v{y} \otimes \v{c} 
	\end{equation} 
	if and only if $F_{\alpha}(\v{x}) > F_{\alpha}(\v{y})$, $\forall \alpha \in \R \backslash\{0\}$ and $F_{\rm Burg}(\v{x}) > F_{\rm Burg}(\v{y})$. $H_C$ can be chosen to be trivial.
\end{thm}

\begin{proof}
	
	

	$F_{\rm \alpha}(\v{x})> F_{\rm \alpha}(\v{y})$ for all $\alpha \in \R \backslash \{0\}$ if and only if $H_{\alpha}(\Gamma_{\v{d}}(\v{x}))< H_{\alpha}(\Gamma_{\v{d}}(\v{y}))$ for all $\alpha \in \R \backslash \{0\}$ follows immediately from Eq.~\eqref{eq:relation}. Furthermore, define the Burg entropy $H_{\rm Burg}(\v{p}):= -1/d \sum_i \log p_i$, where $d$ is the dimension of $\v{p}$. Then,
		\begin{equation}
		H_{\rm Burg}( \Gamma_{\v{d}}(\v{x})) = -S(\v{g}\| \v{x}) + \log D
		\end{equation}
		immediately implies that  $F_{\rm Burg}(\v{x}) > F_{\rm Burg}(\v{y})$ if and only if $H_{\rm Burg}(\Gamma_{\v{d}}(\v{x})) < H_{\rm Burg}( \Gamma_{\v{d}}(\v{y}))$. 
	We also have that, since $\v{x}$, $\v{y}$ have full support and $\v{x} \neq \v{y}$, then  $\Gamma_{\v{d}}(\v{x})$, $\Gamma_{\v{d}}(\v{y})$ have full support and $ \Gamma_{\v{d}}(\v{x}) \neq \Gamma_{\v{d}}(\v{y})$. 
	
	A highly non-trivial result proved by Klimesh and Turgut \cite{klimesh2007inequalities, turgut2007catalytic}
	says that these conditions on the embedded populations are equivalent to the existence of $\v{c}$ such that $ \Gamma_{\v{d}}(\v{x}) \otimes \v{c} \succ  \Gamma_{\v{d}}(\v{y}) \otimes \v{c}$. One can always take $H_C = \mathbb{I}$, so the above is equivalent to $\Gamma_{\v{d}}(\v{x}) \otimes \Gamma_{\v{1}}(\v{c}) \succ  \Gamma_{\v{d}}(\v{y}) \otimes \Gamma_{\v{1}}(\v{c})$, with $\v{1}$ the all ones vector. One can directly verify that $\Gamma_{\v{d}}(\v{x}) \otimes \Gamma_{\v{d}'}(\v{y}) = \Gamma_{\v{d} \otimes \v{d}'}(\v{x} \otimes \v{y})$ for general $\v{x}$,$\v{y}$, $\v{d}$, $\v{d}'$, so the above is equivalent to $\Gamma_{\v{d} \otimes \v{1}}(\v{x}\otimes \v{c}) \succ  \Gamma_{\v{d}\otimes \v{1}}(\v{y}\otimes \v{c})$. From the bridge lemma~\ref{lem:bridge}, this is equivalent to $\v{x} \otimes \v{c} \succ_{g} \v{y} \otimes \v{c}$.  
\end{proof}

From the thermal Nielsen's theorem \ref{thm:thermalnielsen}, these transformations can be approximated arbitrarily well by Thermal Operations. Also note  that $F_{\rm Burg}$ grows unboundedly when a not full rank distribution is approached; in fact, $F_{\rm Burg}$ has been linked to the unattainability of perfect cooling \cite{janzing2000thermodynamic, wilming2017third}. Also, by looking at transformations in which the output is only required to be $\epsilon$-close to the target (with $\epsilon >0$ \emph{arbitrarily} small), one can eliminate any finite number of conditions and make the inequalities non strict \cite{brandao2013second}. 

\begin{rmk}[Asymptotic rates and extensions]
	As in Eq.~\eqref{eq:asymptoticrate}, we can consider the asymptotic limit of a large number $N$ of uncorrelated or weakly correlated particles, which is intuitively analogue to a thermodynamic limit. Then the optimal rate $R$ at which the transformation $\v{x}^{\otimes N}$ to $\v{y}^{RN}$ is possible under Thermal Operations (with negligible error in the $N\rightarrow \infty$ limit) is $R = (kT \log Z + F(\v{x}))/(kT \log Z + F(\v{y}))$ \cite{brandao2011resource}, which extends Remark~\ref{rmk:asymptoticinterconvesionTinfinity} to finite temperatures. The result holds also for general quantum states as input and outputs, under some extra assumptions (the availability of an ancilla with a coherent superposition over a number of energy levels sublinear in $N$ \cite{brandao2011resource}). An extension of this result drops the assumption that thermal states are freely available, and shows that the average energy and the von Neumann entropy are naturally distinguished measures to determine what asymptotic transitions are allowed \cite{sparaciari2017resource} (once again, a `sublinear' number of ancillas are used). Furthermore, corrections to the $N\rightarrow \infty$ rate were worked out in Ref.~\cite{chubb2017beyond}.  
\end{rmk}

\subsubsection{Application: work extraction and work of formation for incoherent states}
\label{sec:workextractionincoherent}

A \emph{deterministic} work extraction process is one in which we are able to charge up a battery system with certainty. The battery can be conveniently modelled as a two-level system with Hamiltonian $H_W = W \ketbra{1}{1}$, initialised in state $ \ketbra{0}{0}_W$ (even though this is not the only choice, see e.g. Appendix I2 of Ref.~\cite{brandao2013second} and Ref.~\cite{gemmer2015from}). Given $\rho_S$ with Hamiltonian $H_S = \sum_{i=1}^n E^S_i \ketbra{E^S_i}{E^S_i}$, the aim is to maximise $W$ such that the transition $\rho_S \otimes \ketbra{0}{0}_W \rightarrow \gamma_S \otimes \ketbra{1}{1}_W$ is allowed by Thermal Operations. Note that we took without loss of generality the final state of the system to be thermal, since one can always thermalise $S$ to such state at the end of the work extraction protocol.

If the initial state $\rho_S$ is diagonal in the energy eigenbasis, with population $\v{x}$, this problem is mapped to a classical one: finding the largest $W$ such that
\begin{equation}
\v{x} \otimes (1,0)_W \succ_g \v{g} \otimes (0,1)_W,
\end{equation}
where $W$ is the energy of the upper state of the battery. Such optimal $W$ is also called \emph{work of distillation} and denoted by $W_{\rm det}$. To compute $W_{\rm det}$ we will make use of the following lemma (see Fig.~\ref{fig:rescaling})):
\begin{lem}
	\label{lem:rescaling}
	If $H_S = \sum_{i=1}^n E^S_i \ketbra{E^S_i}{E^S_i}$ and $H_W = W \ketbra{1}{1}$, for any $\v{y}$ state of $S$ the thermo-majorisation curve of \mbox{$\v{y} \otimes (0,1)_W$} (battery excited) is a compression along the $x$-axis by a factor $e^{-\beta W}$ of the thermo-majorisation curve of $\v{y} \otimes (1,0)_W$ (battery in ground state). 
\end{lem}
\begin{proof}
	 Denote by $\v{g}_W = (g^W_0, g^W_1)$ the thermal state associated to the battery Hamiltonian, where $g^W_0 = (1+e^{-\beta W})^{-1}$, $g^W_1=1-g^W_0$. Hence, the thermal state associated to the Hamiltonian of system+battery is
	\begin{equation}
	\v{g} \otimes \v{g}_W = (g_0^W g_1, \dots ,g_0^W g_n, g_1^W g_1, \dots , g_1^W g_n).
	\end{equation} We have $\v{y} \otimes (1,0)_W = (y_1,...,y_n,0,...,0)$, $\v{y} \otimes (0,1)_W = (0,...,0,y_1,...,y_n)$. If we denote  the energy levels of system and battery by \mbox{$\{(1,0),...,(n,0),(1,1),...,(n,1)\}$}, the permutation $\pi_1$ that $\beta$-orders $\v{y} \otimes (1,0)_W$ has the form $(i,j) \mapsto (\pi(i),j)$, for some permutation $\pi$ (i.e., the $\beta$-order is ($(\pi(1),0), (\pi(2),0), ...$). Furthermore, the permutation $\pi_2$ that $\beta$-orders $\v{y} \otimes (0,1)_W$ has the form $(i,j) \mapsto (\pi(i),\textrm{NOT}(j))$, for the same permutation $\pi$ (i.e., the $\beta$-order is ($(\pi(1), 1), (\pi(2),1),...$). 
	According to Def.~\ref{defn:thermomajorisation}, the $x$-axis points of the thermo-majorisation curve of $\v{y} \otimes (0,1)_W$ are $Z_{SW} \sum_{i=1}^k g^W_1 g_{{\pi}(i)}$ ($k=1,...,n$), where $Z_{SW}$ is the partition function of the Hamiltonian of system+battery, whereas the $x$-axis points of the thermo-majorisation curve of $\v{y} \otimes (1,0)_W$ are $Z_{SW} \sum_{i=1}^k g^W_0 g_{{\pi}(i)}$ $(k=1,...,n)$. The corresponding $y$-axis coordinates instead coincide, in both cases being equal to $\sum_{i=1}^k y_{{\pi}(i)}$.  In other words, 
	the two curves are the same apart from a overall rescaling of the $x$-axis by a factor $e^{-\beta W} = g^W_1/g^W_0$.
\end{proof}

\begin{figure}[h!]
	\centering
	\includegraphics[width=0.5\columnwidth]{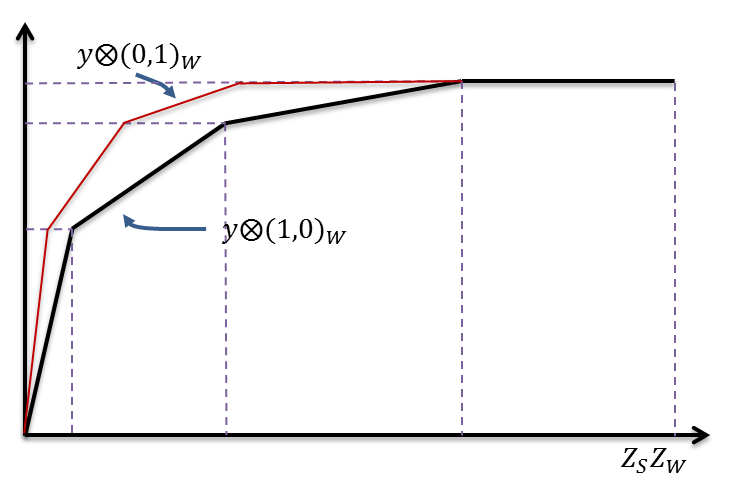}
	\caption{\label{fig:rescaling} {\bf Work rescales the thermo-majorisation curve:} the thermo-majorisation curve of $\v{y} \otimes (0,1)_W$ (battery excited) is a compression along the $x$-axis of the thermo-majorisation curve of $\v{y} \otimes (1,0)_W$ (battery in ground state) by a factor $g^W_1/g^W_0 =e^{-\beta W}$ (in this example we took this ratio to be $1/2$).  }
\end{figure}

One has, recalling the definition of Eq.~\eqref{eq:zerorenyi},
\begin{corol}[Deterministic work extraction \cite{horodecki2013fundamental, aberg2013truly}]
\begin{equation}
\label{eq:singleshotwork}
W_{\rm det} = - k T \log \sum_{i|x_i \neq 0} g_{i}  = k T S_0(\v{x}\| \v{g}) = F_0(\v{x}) - F_0(\v{g}).
\end{equation}
\end{corol}
\begin{proof}

Following Fig.~\ref{fig:detwork}, we will construct the thermo-majorisation curves $T(\v{g}\otimes(0,1)_W)$ for varying $W$ and look for the largest $W$ such that this curve is all below $T(\v{x}\otimes(1,0)_W)$.

To construct $T(\v{g}\otimes(0,1)_W)$ let us start with $T(\v{g}\otimes(1,0)_W)$. This has constant slope $1/Z_S$ in $x \in [0,Z_S]$, where $Z_S =\sum_i e^{-\beta E_i}$, and is flat in $x\in [Z_S,Z_{SW}]$, where $Z_{SW} = Z_S Z_W$, with $Z_W = 1+ e^{-\beta W}$ (orange dashed line in Fig.~\ref{fig:detwork}). By Lemma~\ref{lem:rescaling}, the thermo-majorisation curve $T(\v{g} \otimes (0,1)_W)$ has slope $e^{\beta W}/ Z_S$ in $x \in [0, e^{-\beta W} Z_S]$, and is flat in $x \in [e^{-\beta W} Z_S, Z_{SW}]$ (so the part that is not flat connects the origin and $(e^{-\beta W} Z_S,1)$).

To compare  $T(\v{g}\otimes(0,1)_W)$ to $T(\v{x}\otimes(1,0)_W)$, we look for the point at which the latter curve reaches height~$1$. Let $\pi$ be the permutation that $\beta$-orders $\v{x}$ and let $k$ be the smallest number such that $\sum_{i=1}^k x_{\pi(i)} =1$. The $x$-axis point at which $T(\v{x}\otimes(1,0)_W)$ reaches height $1$ is then $Z_S \sum_{i=1}^k g_{\pi(i)}$. Comparing with the family of curves $T(\v{g} \otimes (0,1)_W)$ constructed before, we see that the largest $W$ such that $T(\v{g} \otimes (0,1)_W)$ is not above $T(\v{x}\otimes(1,0)_W)$ is the $W$ satisfying $Z_S e^{-\beta W} = Z_S \sum_{i=1}^k g_{\pi(i)}$, giving  $W= W_{\rm det} = - k T \log \sum_{i=1}^k g_{\pi(i)}$ (blue dotted line in Fig.~\ref{fig:detwork}). In fact, we chose $W=W_{\rm det}$ so that the elbow of $T(\v{g} \otimes (0,1)_W)$ lies on  $T(\v{x}\otimes(1,0)_W)$; furthermore, $T(\v{g} \otimes (0,1)_W)$ is a straight line from the origin to the elbow; hence, from the concavity of thermo-majorisation curves, we conclude $T(\v{g} \otimes (0,1)_W)$ is all below $T(\v{x}\otimes(1,0)_W)$). The result in its final form of Eq.~\eqref{eq:singleshotwork}, follows from the definition of $k$.
\begin{figure}[h!]
	\centering
	\includegraphics[width=0.5\columnwidth]{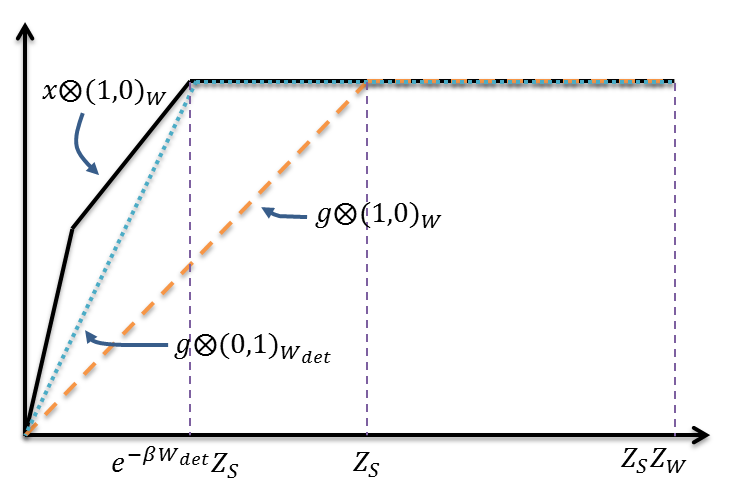}
	\caption[Deterministic work extraction]{\label{fig:detwork} {\bf Deterministic work extraction:} the thermo-majorisation curve of the system state $\v{x}$ with de-excited battery $(1,0)_W$ is in black; the thermal state $\v{g}$ with de-excited battery $(1,0)_W$ is represented by the dashed orange curve; exciting the battery corresponds to compressing the corresponding curve by a factor $e^{-\beta W }$ along the $x$-axis (see Lemma~\ref{lem:rescaling}). The dotted blue curve corresponds to the thermal state with the most excited battery state possible, under the condition that the curve lies all below the black curve of the initial state.}  
\end{figure}
\end{proof}

 Note the role the single-shot quantity $S_0(\v{x}\| \v{g})$ plays in characterising the work extractable deterministically. On average the largest of amount of work extractable from $\v{x}$ using a bath at temperature $T$ is $W_{\rm ave} = kT S_1(\v{x}\| \v{g}) > W_{\rm det}$ (for a formal treatment, see Ref.~\cite{aberg2013truly}). Also note that no deterministic work can be extracted from states with full support. Extensions allowing for some $\epsilon$ probability of failure have been formulated \cite{aberg2013truly,horodecki2013fundamental}.

A question related to the above is what is the minimum amount of work necessary to create a state, something called the \emph{work of formation} \cite{horodecki2013fundamental}. This is defined as the minimum amount of work $W_{\rm for}$ necessary to create a quantum state $\rho_S$ from the thermal state $\gamma_S$ under Thermal Operations: $
\gamma_S \otimes \ketbra{1}{1}_W \rightarrow \rho_S \otimes \ketbra{0}{0}$.
 For diagonal target states, this problem reduces to finding the smallest $W$ such that 
 \begin{equation}
 \v{g} \otimes (0,1)_W \rightarrow \v{x} \otimes (1,0)_W.
 \end{equation}
 Using the same reasoning as above, based around Lemma~\ref{lem:rescaling} (see Fig.~\ref{fig:workformation}), one can see that it is necessary and sufficient to add an amount of work $W$ that makes the slope of $T(\v{g} \otimes (0,1)_W)$ larger than the biggest slope in $T(\v{x} \otimes (1,0)_W)$. Since, as we described before, the slope of $T(\v{g} \otimes (0,1)_W)$ is $e^{\beta W} \frac{1}{Z_S}$, and the slopes of the segments in $T(\v{x} \otimes (1,0)_W)$ are $\frac{x_i}{e^{-\beta E^S_i}}$, that means
 \begin{equation}
 e^{\beta W_{\rm for}} \frac{1}{Z_S} = \max_i \frac{x_i}{e^{-\beta E^S_i}} \Rightarrow W_{\rm for} = kT S_\infty(\v{x}\|\v{g}) = F_\infty(\v{x}) - F_\infty(\v{g}).  
 \end{equation}
 Note that since $S_0(\v{x}) < S_\infty(\v{x})$ for every non-thermal distribution, once $\v{x}$ is created expending $W_{\rm for}$ only a smaller amount $W_{\rm det} $ can be extracted from it, i.e. the cycle $\v{g} \rightarrow \v{x} \rightarrow \v{g}$ is irreversible.
 
 \begin{figure}[h!]
 	\centering
 	\includegraphics[width=0.5\columnwidth]{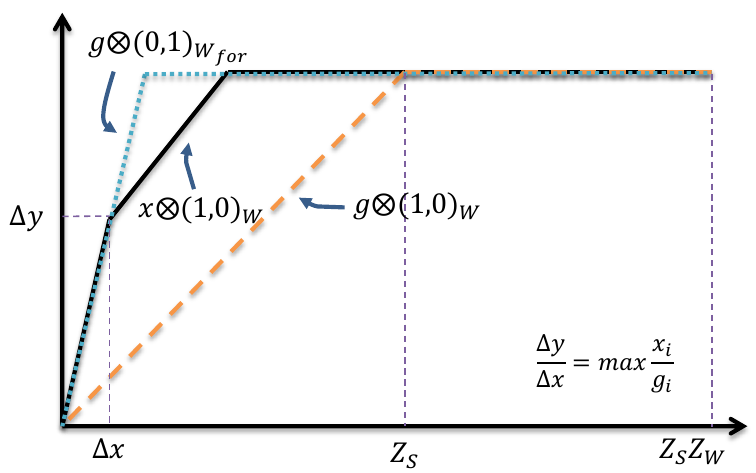}
 	\caption{ {\bf Work of formation:} the state $\v{x}$ with de-excited battery $(1,0)_W$ is presented in black. The initial state $\v{g}$ with de-excited battery $(1,0)_W$ is given by the dashed orange curve. The state $\v{g}$ with charged battery is represented by a blue dotted curve, which is a compression by $e^{-\beta W_{ \rm for}}$ of the orange curve (see Lemma~\ref{lem:rescaling}). $W_{ \rm for}$ is the value of work that takes the orange dashed curve all above the black curve, i.e. the minimum amount of work that needs to be consumed if we wish to create $\v{x}$ by discharging the battery. }	\label{fig:workformation} 
 \end{figure}

We will discuss these questions for arbitrary states once we introduce thermodynamic constraints on the evolution of quantum coherence. This is the next topic we will consider.
  
  \section{Thermodynamic laws for coherence}
  \label{sec:coherence}
  
    While in the classical scenario the second law, in its generalized thermo-majorisation form, only constrains the allowed population dynamics, we are now interested in understanding the thermodynamic processing of quantum coherence (unless otherwise stated, we always refer to coherence in the energy basis). A generic non-equilibrium initial state can be found in some superposition of energy states, such as \mbox{$\ket{\psi} = (\ket{0} + \ket{1})/\sqrt{2}$}. The occupations of ground and excited states are here $\v{x} = (1/2,1/2)$ and, as we know from the thermal Nielsen's theorem, under Thermal Operations $\v{x}$ will ``approach'' $\v{g}$, in the sense that $\v{y}$ is an achievable final population if and only if \mbox{$\v{x} \succ_g \v{y}$}. At the same time, however, $\ket{\psi}$ carries a superposition of energy eigenstates with amplitude $|c|= |\! \braket{0}{\psi}\! \braket{\psi}{1}\!| = 1/2$ that, intuitively, will get degraded due to decoherence. So, what are the achievable amplitudes $|c'|$ in the final state, given a transition $\v{x} \rightarrow \v{y}$ in the diagonal? We want to formalise this into explicit constraints on the decay of quantum coherence. For example,
    \begin{equation}
    \rho_S = \begin{pmatrix}
    x_1 & \rho_{01} & \rho_{02} \\
    \rho_{10}  & x_2 & \rho_{12} \\
    \rho_{20} & \rho_{21} & x_3
    \end{pmatrix} \quad \longrightarrow \quad  \sigma_S = \begin{pmatrix}
    y_1 & ? & ? \\
    ?  & y_2 & ? \\
    ? & ? & y_3
    \end{pmatrix}
    \end{equation}
    
 In other words, we need to go beyond thermo-majorisation. To see why thermo-majorisation together with positivity of the quantum state is insufficient, consider the transformation
 \begin{equation}
 \gamma_S :=e^{-\beta H_S}/Z_S \rightarrow \ket{\gamma} := \sum_i \sqrt{g_i} \ket{E^S_i}, \quad H_S = \sum_i E^S_i \ketbra{E^S_i}{E^S_i}. 
 \end{equation} 
 It should be obvious that both $\gamma_S$ and $\ket{\gamma}$ are associated to the population vector $\v{g}$, so that the thermo-majorisation condition is trivially satisfied. On the other hand, one can verify from Eq.~\eqref{eq:thermal} that $\mathcal{T}(\gamma_S)= \gamma_S$, hence there is no Thermal Operation mapping $\gamma_S$ into $\ket{\gamma}$. In fact, as we will see, in a precise sense thermo-majorisation is a `zero mode' constraint of an entire hierarchy of thermodynamic relations.

\subsection{Time-translation symmetry and thermodynamics}   
While it is intuitive from the previous analysis that it makes sense to consider the population and the coherent components of quantum states separately, this distinction is not refined enough. The following considerations are based on a symmetry analysis of Thermal Operations. The tools used allow to deal with symmetries in open quantum systems, i.e. quantum systems interacting with an external environment; as such, they can be understood as an extension of Noether's theorem to open evolutions \cite{marvian2014extending}. In particular, one exploits a harmonic analysis of  quantum states \cite{marvian2014modes} that brings to the fore structures implicit in standard treatments \cite{breuer2002open}. We introduce the necessary considerations, in an elementary fashion, in the next section. We recommend Ref.~\cite{marvianthesis} for  details of the symmetry analysis.  

\subsubsection{Extending Noether's theorem to open systems}

  A symmetry group $G$ acts on the set of density matrices $\rho$ through the following representation:
  \begin{equation}
  \label{eq:representation}
  	g \in G \mapsto \mathcal{U}_g(\cdot) = U_g(\cdot)U^\dag_g
  \end{equation}
where $U_g$ is a unitary. 
For example, $G=U(1)$ (or $G= \mathbb{R}$) is the group generated by the Hamiltonian $H_S$, $U_t= e^{-i H_S t}$, or rotations about an axis. Another common example is $G=SU(2)$. A closed system dynamics $V$ is said to be symmetric when it commutes with the action of the group, $[V, U_g] = 0$ for all $g \in G$. If $G$ is a Lie group, Noether's theorem implies that, if a closed system  dynamics exists mapping $\rho_S$ into $\sigma_S$, the generators of $G$ and all their powers (the Lie algebra of $G$)  are conserved quantities: $\tr{}{\rho_S H_S^k} = \tr{}{\sigma_S H_S^k} \, \forall k \in \mathbb{N}$, for the case of $G= U(1)$ generated by $H_S$. 

It should be clear that open systems present further difficulties; in particular, an open system dynamics in general has no conserved quantity, even if a conservation law holds at the level of system+environment. However, a notion of symmetry can be naturally defined for open dynamics, and turns out to be directly related to such global conservation laws. Define
\begin{defn}[Symmetry of open dynamics]
	Given a group $G$ with a representation as in Eq.~\eqref{eq:representation}, a channel $\mathcal{E}$ is \emph{symmetric with respect to $G$}, or $G$-\emph{covariant}, if $[\mathcal{E}, \mathcal{U}_g] = 0$ for all $g \in G$, i.e. $\mathcal{E}[\mathcal{U}_g(\rho_S)] = \mathcal{U}_g[\mathcal{E}(\rho_S)]$ for every $\rho_S$ and every $g \in G$. 
\end{defn}
In fact, one can construct a resource theory in which the set of free operations are those symmetric with respect to $G$ \cite{marvian2013asymmetry}. This is a theory of quantum coherence between eigenspaces of the observables generating $G$. If the generator is the Hamiltonian $H_S$, $\mathcal{E}$ is said to be \emph{time-translation symmetric}, also known as phase covariant or phase insensitive channels.

  While symmetries of open quantum systems do not in general imply conservation laws, they imply that certain quantities are monotonically decreasing under symmetric operations. These are called \emph{asymmetry monotones} and they are functions that capture aspects of the partial ordering induced on quantum states by the set of symmetric operations:
  \begin{defn}[Asymmetry monotone]
  	A functional $a$ is called asymmetry monotone for $G$ if \begin{equation*}
  	 a(\mathcal{E}(\rho_S)) \leq a(\rho_S)
  	\end{equation*}
  	for every $\rho_S$ and every $G$-covariant channel $\mathcal{E}$.
  \end{defn}
We define a state $\rho_S$ to be \emph{symmetric} if $\mathcal{U}_g(\rho_S) = \rho_S$ for all $g \in G$. These states can only contain incoherent mixtures of distinct eigenstates of generators of $G$.

 \begin{rmk}[Different notions of quantum coherence]
	The above mentioned notion of quantum coherence is one in which the particular encoding is relevant: if $H_S = \ketbra{1}{1} + 2\ketbra{2}{2}$, the states $(\ket{0}+\ket{1})/\sqrt{2}$ and $(\ket{0} + \ket{2})/\sqrt{2}$ behave differently under time translations, despite being both equal superpositions of two eigenstates in a preferred basis. Such notion of quantum coherence has been dubbed \emph{unspeakable} \cite{bartlett2007reference, marvian2016how} (since the labels `$0$', `$1$', `$2$' have a physical meaning, e.g. as eigenstates of the Hamiltonian) and is the relevant one for thermodynamics, metrology and quantum speed limits among other things. It is to be contrasted with a more computational notion of quantum coherence \cite{streltsov2017colloquium} in which the two states above are to be considered equivalent (since the labels `$0$', `$1$', `$2$', are irrelevant). The latter notion is termed \emph{speakable} quantum coherence and does not appear to capture quantum thermodynamic constraints. For a more detailed discussion see Ref.~\cite{marvian2016how}. 
\end{rmk}

Qualitatively, if $\rho_S$ is not symmetric, the application of a symmetric evolution $\mathcal{E}$ will make it `more symmetric', i.e. bring it closer to the set of states $\sigma_S$ satisfying $\mathcal{U}_g(\sigma_S) = \sigma_S$ for all $g \in G$. Asymmetry monotones make this statement quantitative. Let us present an example of such a quantity, whose thermodynamic relevance will be clarified later. Define
  \be
  \label{eq:Gtwirling}
  \mathcal{G}(\rho_S)= \int_G \mathcal{U}_g(\rho_S) dg,
  \ee
 as the average over all group elements ($dg$ being the Haar measure associated to $G$, assuming it exists). The operation $\mathcal{G}$ is known as $G$-\emph{twirling}. For $G=U(1)$ generated by $H_S$, $\mathcal{G}$ corresponds to the dephasing operation $\mathcal{D}$ in Eq.~\eqref{eq:dephasing}.
 Then define
  
  \begin{defn}[Asymmetry]
  	\emph{Asymmetry} is the asymmetry monotone defined as
  	\begin{equation}
  	\label{eq:asymmetry}
  	A(\rho_S) = S(\mathcal{G}(\rho_S)) - S(\rho_S) = S(\rho_S\|\mathcal{G}(\rho_S)),
  	\end{equation}
  	where the \emph{relative entropy} is $S(\rho_S\|\sigma_S) = \tr{}{\rho_S (\log \rho_S - \log \sigma_S}$ and $S(X) = - \tr{}{X\log X}$ is the von Neumann entropy. 
   \end{defn}
The two expressions in Eq.~\eqref{eq:asymmetry} coincide because $\tr{}{\rho_S \log \mathcal{G}(\rho_S)} =  \tr{}{\mathcal{G}(\rho_S) \log \mathcal{G}(\rho_S)}$ (use $\mathcal{G} \circ \mathcal{G} = \mathcal{G}$). Using the contractivity of the relative entropy ($S(\mathcal{E}(\rho_S)\| \mathcal{E}(\sigma_S)) \leq S(\rho_S\| \sigma_S)$ for all channels $\mathcal{E}$) and $\mathcal{E}\circ \mathcal{G} = \mathcal{G}\circ \mathcal{E}$ if $\mathcal{E}$ is $G$-covariant, one can immediately derive $A(\mathcal{E}(\rho_S)) \leq A(\rho_S)$. Asymmetry monotones replace conservation laws for open systems \cite{marvian2014extending}.
  
  Even in closed system dynamics these considerations are relevant. In fact, conservations laws on the generators of $G$ are insufficient to characterise what mixed state transformations are possible under closed symmetric evolutions, as the following example shows:
  	\begin{ex}[Asymmetry monotones are necessary even in closed systems \cite{marvian2014extending}]
  	 Consider a system described by $\mathcal{H}_S \otimes \mathcal{H}_A$ where $\mathcal{H}_S$ is a qubit system and $\mathcal{H}_A$ is an ancilla. Then define the two states
  	\be
  	\nonumber
  	\rho_{SA} = \frac{1}{2} \ketbra{0}{0} \otimes \ketbra{s_1}{s_1} + \frac{1}{2} \ketbra{1}{1} \otimes \ketbra{s_2}{s_2}, \quad \xi_{SA}= \frac{1}{2} \ketbra{+}{+} \otimes \ketbra{s_1}{s_1} + \frac{1}{2} \ketbra{-}{-} \otimes \ketbra{s_2}{s_2},
  	\ee
  	where $\ket{0}$, $\ket{1}$ are eigenstates of the Pauli $Z$ operator and $\ket{\pm}$ are eigenstates of the Pauli $X$ operator. A unitary exists mapping $\rho_{SA}$ into $\xi_{SA}$ (Hadarmard on the first system). However, assume we can only perform rotationally symmetric dynamics ($SU(2)$-covariant unitaries), and that $\ket{s_1}$, $\ket{s_2}$ are two orthogonal states of a set of degrees of freedom invariant under rotations, so that rotations act trivially on $\mathcal{H}_A := {\rm span} \{\ket{s_1},\ket{s_2}\}$. Is it possible to find a symmetric unitary dynamics transforming $\rho_{SA}$ into $\xi_{SA}$? The generators of the symmetry are $\sigma_i \otimes \mathbb{I}_A$, where $\sigma_1= X$, $\sigma_2 = Y$, $\sigma_3 = Z$. Since the reduced state on the first system is maximally mixed for both $\rho_{SA}$ and $\xi_{SA}$, one finds $
  	\tr{}{\rho_{SA} \sigma_i \otimes \mathbb{I}_A} = \tr{}{ \xi_{SA} \sigma_i \otimes \mathbb{I}_A}$ for $i=1,2,3$. 
  	So all generators of the symmetry group acting on $\mathcal{H}_S \otimes \mathcal{H}_A$ are conserved quantities. Nevertheless, there is no symmetric transformation (unitary or otherwise) mapping $\rho_{SA}$ into $\xi_{SA}$. This is easily captured by asymmetry monotones. To see this, one needs to generalize $A$ to the Holevo asymmetry monotone $A_p := S(\mathcal{G}_p(\rho_S)) - S(\rho_S)$, where $\mathcal{G}_p(\rho_S)= \int_G p(g) \mathcal{U}_g(\rho_S) dg$ for any probability density $p(g)$ over $G$. Taking $p(g)$ to be uniform on the $U(1)$ subgroup of $SU(2)$ generated by $Z$ and zero otherwise, $\mathcal{G}_p$ becomes an average over all rotations about the $z$ axis on the $S$ system. Hence
  	\be
  	\mathcal{G}_p(\xi_{SA}) = \frac{\I}{4} \otimes \ketbra{s_1}{s_1} + \frac{\I}{4} \otimes \ketbra{s_2}{s_2} = \frac{\I}{4} , \quad 	\mathcal{G}_p(\rho_{SA}) = \rho_{SA},
  	\ee
  	from which we obtain $A_p(\xi_{SA}) = S\left([1/2,0,0,1/2]\|[1/4,1/4,1/4,1/4]\right) = \log 2$, while $A_p(\rho_{SA}) = 0$. No symmetric dynamics exists mapping $\rho_{SA}$ into $\xi_{SA}$, since \mbox{$ A_p(\xi_{SA})>A_p(\rho_{SA})=0$}. 
  \end{ex}
 
 The following theorem (that we give without proof) provides a dilation of channels with a $U(1)$ symmetry. As suggested above, when we gave the definition of symmetric channel, the dilation shows that covariance can be understood as arising from conservation laws on an enlarged system. Also symmetries can `go to the church of the larger Hilbert space':
 \begin{thm}[Stinespring dilation for time-translation symmetric maps]
 	\label{thm:stinespringcovariant}
 	Suppose $S$ has Hamiltonian $H_S$ and $\mathcal{E}$ is a time-translation symmetric channel on $S$. Then there exists an ancillary system $\sigma_A$ with Hamiltonian $H_A$ and a unitary $U$ on $SA$ such that $[\sigma_A, H_A] =0$, $[U, H_S + H_A] = 0$ and
 	\begin{equation}
 	\mathcal{E}(\rho_S) = \tr{A}{U(\rho_S \otimes \sigma_A)U^\dag}.
 	\end{equation}
 \end{thm}
The result holds for more general symmetry groups $G$, see Sec.~4.4 of Ref.~\cite{marvianthesis}, and Appendix~B of Ref.~\cite{keyl1999optimal}. In the final example we discuss the relation between time-translation symmetry and standard approximations performed in the context of open quantum system dynamics, which may be useful to those familiar with the latter:
 
  \begin{rmk}[Master equations and $U(1)$-covariance \cite{lostaglio2017markovian}]
  	\label{rmk:masterequations}
 	Consider the set of channels $\mathcal{E}$ that admit a time-independent generator $\mathcal{L}$, meaning that there exists $s>0$ and a Lindbladian $\mathcal{L}$ such that $\mathcal{E} = e^{\mathcal{L} s}$ (see Ref.~\cite{breuer2002open}, Section~3.2; the jargon is that $\mathcal{E}$ is time-homogeneous Markovian). In standard microscopic derivations of master equations one performs the secular or rotating wave approximation after the Born-Markov approximation (typically justified in the weak coupling limit, see Section~3.3 of Ref.~\cite{breuer2002open}). This ensures that $\mathcal{L}$ commutes with the superoperator $\mathcal{H}:= [H, \cdot]$ that generates the unitary part of the dynamics. A direct calculation shows that the resulting channel $\mathcal{E}$ is time-translation symmetric with respect to the group generated by $H$.
 	 This provides a point of view on the emergence of time-translation symmetry in practical considerations that is rooted in the master equation formalism. In fact, a typical set of channels used to study thermodynamic processes in the weak coupling limit are the so-called Davies maps \cite{davies1974markovian}, which are examples of time-translation symmetric channels.
 \end{rmk}

\subsubsection{Time-translation symmetry of Thermal Operations}
\label{sec:timetranslationsymmetryTO}
  
  Consider the action of time translations (a $U(1)$ group generated by $H_S$) on the set of quantum states: $t \mapsto \mathcal{U}_t(\cdot) = e^{-i H_S t}(\cdot)e^{i H_S t}$. The initial states $\rho_S$ for which thermo-majorisation gives necessary and sufficient conditions are, as we have seen, those for which $[\rho_S, H_S] =0$, i.e. with no coherence \emph{among} energy eigenspaces, or \emph{incoherent} for short (see Theorem~\ref{thm:thermalnielsen} and Remark~\ref{rmk:simplextension}). This can be equivalently written as $\mathcal{U}_t(\rho_S) = \rho_S$ for all $t$, i.e. states that are incoherent in the energy basis are those that are \emph{symmetric} under the action of time-translations. Only for those thermo-majorisation is the whole story.
  
  Now consider the action of a Thermal Operation $\mathcal{T}$ on a time translated state $\mathcal{U}_t(\rho_S)$. Using the invariance of $\gamma_B$ under the time translations generated by $H_B$ and the commutation relation $[U,H_S+H_B] = 0$, from Eq.~\eqref{eq:thermal} one can see that, as anticipated in Remark~\ref{rmk:alternative},
  \begin{equation}
  \label{eq:covariant}
  \mathcal{T}(\mathcal{U}_t(\rho_S)) = \mathcal{U}_t(\mathcal{T}(\rho_S)) \quad \forall t \, \; \; \forall \rho_S.
  \end{equation} 
 We conclude that Thermal Operations are $U(1)$-\emph{covariant} or \emph{time-translation symmetric} \cite{lostaglio2015description}. Physically this tells us that it does not matter if we apply $\mathcal{T}$ at time $s=0$ and then let the system freely evolve for some time $t$ or we invert the order of the operations: the final state will be identical. It also tells us that Thermal Operations do not require any external source of coherence. This is obvious from the definition in Eq.~\eqref{eq:thermal} (since $\gamma_B$ is an incoherent state). However, it follows from Eq.~\eqref{eq:covariant} alone, as Theorem~\ref{thm:stinespringcovariant} shows. Since symmetric evolutions can only degrade asymmetry properties, and in our case asymmetry coincides with energy coherence, we see that the fact that Thermal Operations are symmetric implies that they degrade quantum coherence. Hence, $A(\rho_S)$ is one measure of coherence that needs to decrease under Thermal Operations. Another one is the quantum Fisher information with respect to the unitary orbit generated by $H_S$:

\begin{ex}[Quantum Fisher Information degradation under Thermal Operations \cite{janzing2003quasi}]
	\label{ex:fisher}
The quantum Fisher information for the family $\{\mathcal{U}_t(\rho_S)\}$ is an asymmetry monotone. Let $Q(\rho_S,t)$ be defined as 
	\begin{equation}
	Q(\rho_S,t) := 2 \lim_{\delta \rightarrow 0} (1-\mathcal{F}(\rho_t,\rho_{t+\delta})^2)/\delta^2
	\end{equation}
	where $\mathcal{F}(\cdot,\cdot)$ is the fidelity, $\mathcal{F}(\rho,\sigma)= \tr{}{\sqrt{\rho^{1/2} \sigma \rho^{1/2}}}$, and $\rho_t := \mathcal{U}_t (\rho_S)$. Now, if $\mathcal{E}$ is a symmetric channel, 
	\begin{equation}
	\mathcal{F}(\mathcal{E}(\rho_S)_t,\mathcal{E}(\rho_S)_{t+\delta}) = \mathcal{F}(\mathcal{E}(\rho_t),\mathcal{E}(\rho_{t+\delta})) \geq \mathcal{F}(\rho_t,\rho_{t+\delta}),
	\end{equation}
	where in the first step we used the condition that $\mathcal{E}$ is symmetric and in the second that $\mathcal{F}$ is contractive under quantum channels  (see Section 3.2 of Ref.~\cite{watrous2018theory}). The above implies that under any symmetric channel \mbox{$Q(\mathcal{E}(\rho_S),t) \leq Q(\rho_S,t)$}, i.e. the quantum Fisher information is an asymmetry monotone and hence, in particular, it is a Thermal Operations monotone: for every $\rho_S$ and Thermal Operation $\mathcal{T}$, $Q(\mathcal{T}(\rho_S),t) \leq Q(\rho_S,t)$.
\end{ex}

 The considerations above lead to a point of view on Thermal Operations as a set of maps that satisfy two core properties:
 \begin{enumerate}
 	\item \label{eq:GPcondition} $\mathcal{T}(\gamma_S) = \gamma_S$, the \emph{Gibbs-preserving condition}, ensures that no external work can be brought in for free (we want a fair accounting of the work resources employed).
 	\item \label{eq:covariancecondition} $\mathcal{T} \circ \mathcal{U}_t = \mathcal{U}_t \circ \mathcal{T} \; \; \forall t$, the \emph{symmetry condition}, ensures that no external source of coherence can be brought in for free (we want a fair accounting of coherent resources, see also considerations in Remark~\ref{rmk:alternative}).
 \end{enumerate}
In fact, as anticipated in Remark~\ref{rmk:alternative}, it is tempting to focus on the superset of Thermal Operations that satisfy these two properties, the Thermal Processes. While these allow the same set of transformations among states as Thermal Operations on qubit systems, the situation is unclear in higher dimension~\cite{cwiklinski2015limitations}. Necessary and sufficient, but implicit, conditions for transformations to be possible under this set were given in Ref.~\cite{gour2017quantum2}. We leave the following conjecture open: 
\begin{rmk}
	[Conjecture] The closure of the set of states achievable with Thermal Operations coincides with the set of states achievable with Thermal Processes. If true, this would be a remarkable simplification of the set of operations we need to consider (not least, the question of the existence of a Thermal Operation $\rho_S \rightarrow \sigma_S$ would be proven to be a semidefinite program, and the results of Ref.~\cite{gour2017quantum2} would be applicable to Thermal Operations). If false, it would mean there is more to Thermal Operations than the two core properties listed  above, and it would be interesting to understand the physical meaning of the extra constraints.
	\label{conjecture} 
\end{rmk}

 The symmetry constraints introduce `second laws for coherence', i.e. analogues of Eq.~\eqref{eq:secondlaws} for quantum coherence. An example of such relations can be obtained by introducing the quantum R\'enyi divergences (see Ref.~\cite{mosonyi2015quantum} and references therein)
\be
\nonumber
S_{\alpha}(\rho_S\|\sigma_S) =
\begin{cases} 
	\frac{1}{\alpha -1}\log \tr{}{\rho_S^{\alpha} \sigma_S^{1-\alpha}}, \; \, \quad \quad \quad \quad \, \quad \alpha \in (0,1), \\
	\frac{1}{\alpha-1}\log \tr{}{\left(\sigma_S^{\frac{1-\alpha}{2\alpha}}\rho_S \sigma_S^{\frac{1-\alpha}{2\alpha}}\right)^\alpha}, \quad \quad \alpha >1.
\end{cases}
\ee 
The limit for $\alpha \rightarrow 1$ is given by $S_1(\rho_S\|\sigma_S) = \tr{}{\rho_S (\log \rho_S - \log \sigma_S)} = S(\rho_S \| \sigma_S)$. Also $\alpha \rightarrow 0, \infty$ are defined by suitable limits: denoting by $\Pi_{\rho_S}$ the projector on the support of $\rho_{S}$,
\be
S_0(\rho_S\|\sigma_S)= - \log \Pi_{\rho_S} \sigma_S, \quad S_\infty(\rho_S\|\sigma_S) = \log \min\{ \lambda: \rho_S \leq \lambda \sigma_S\}.
\ee
These quantities have the (non-obvious!) property of being contractive under quantum channels ($S_\alpha (\mathcal{E}(\rho_S)\|\mathcal{E}(\sigma_S)) \leq S_{\alpha} (\rho_S\|\sigma_S) $ for every $\alpha \geq 0$ and every channel $\mathcal{E}$). 
Then, one can define for any $\alpha \geq 0$
\be
\nonumber
A_{\alpha}(\rho_S):=S_{\alpha}(\rho_S \| \mathcal{D}(\rho_S)),
\ee
which recovers asymmetry for $\alpha =1$. Since $[\mathcal{E},\mathcal{U}_t]=0$ for every $t$ and $\mathcal{D}(\cdot) = \int dt \mathcal{U}_t(\cdot)$, it is simple to show that $[\mathcal{E},\mathcal{D}] =0$. This, together with the contractivity of the $\alpha$-relative entropy, immediately implies that under any Thermal Operation (see Fig.~\ref{fig:blob})
\begin{equation}
\label{eq:alphaasymmetry}
\Delta A_\alpha (\rho_S) \leq 0 \quad \forall \alpha \geq 0.
\end{equation}

That the $A_\alpha$ constraints, together with thermo-majorisation, cannot be sufficient to characterise Thermal Operations follows from the fact that we proved $\Delta A_\alpha \leq 0$ using only the property $[\mathcal{E}, \mathcal{D}] = 0$, which defines a strict superset of time-translation symmetric channels \cite{marvian2016how}. The reason is that Thermal Operations operate independently on different `coherence modes' of the quantum state, as we will discuss in Sec.~\ref{subsection:modes}.

\begin{figure}[h]
	\begin{center}
		\includegraphics[width=0.45\textwidth]{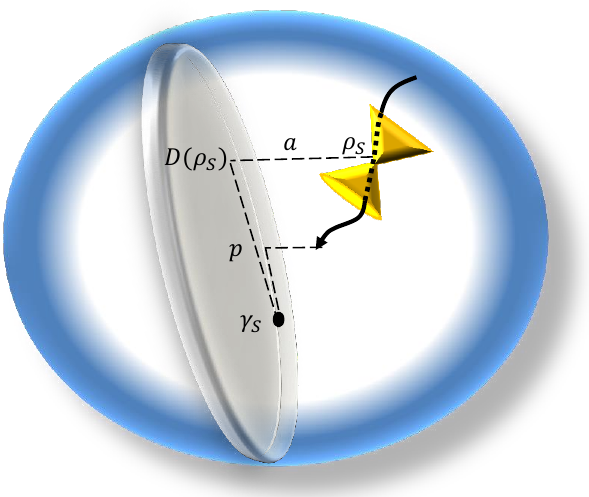}
	\end{center}
	\caption[Quantum Thermodynamics as a theory of athermality and coherence]{{\bf Quantum thermodynamics as a theory of athermality and quantum coherence.} The blue blob pictorially represents the convex set of all quantum states. Any state $\rho_S$ is associated to a ``thermal cone'' (in yellow), the convex set of states accessible from it by means of Thermal Operations (and the backward cone of states that can access it), with a representative trajectory. For any state $\rho_S$ we can identify measures $p$ of its athermality -- which corresponds to the deviation of $\mathcal{D}(\rho_S)$ from the thermal state $\gamma_S$, as measured by $\{F_\alpha \}$ and other thermodynamic Schur-concave functions (see Sec.~\ref{sec:thermodynamicschur}); and asymmetry $a$ -- which corresponds to the deviation, as measured by $\{A_\alpha\}$ or any other asymmetry monotone, of $\rho_S$ from the manifold of time-symmetric states (grey flat region). All of these must be monotonic during the thermalisation process.}
	\label{fig:blob}
\end{figure}

\subsubsection{Coherence constraints are not reducible to free energies. Coherent and incoherent components of the free energy}
\label{sec:decomposition}

It is important to recognize that the constraints imposed by time translation symmetry are not reducible to standard considerations involving \emph{free energy measures}. The intuition is as follows. Let $\Delta \tilde{F}$ be any of the free energy difference measures introduced in the literature. These include the quantum free energy difference \begin{equation}
\Delta F(\rho_S) = kT S(\rho_S \| \gamma_S) = F(\rho_S) - F(\gamma_S), \quad F(X_S) = \tr{}{X_S H_S} - kT S(X_S),
\end{equation}  
where $S(X_S)$ is the von Neumann entropy, the $\alpha$-free energies defined as $\Delta F_{\alpha}(\rho_S) := k T S_{\alpha}(\rho_S\|\gamma_S)$ where $S_\alpha(\cdot\|\cdot)$ are the $\alpha$-R\'enyi divergences defined above (for incoherent states, $\Delta F_\alpha(\rho_S) = F_\alpha(\v{x}) - F_\alpha(\v{g})$, with $\v{x}$ the eigenvalues of $\rho_S$ and $F_\alpha(\v{x})$ defined in Eq.~\eqref{eq:alphafreeenergy}). 
These quantities, and arguably all measures that can be meaningfully called free energies, have the property that they are finite, at least for full rank states, and grow unboundedly on pure energy states $\ket{E}$ of increasing energy, \mbox{$\Delta\tilde{F} (\ket{E}) \rightarrow \infty$} as $E \rightarrow \infty$. However note that, for any $\epsilon>0$, the transformation \mbox{$\ket{E} \rightarrow \sigma^\epsilon_S$}, with \mbox{$\sigma^\epsilon_S:= \epsilon \ketbra{+}{+} +(1-\epsilon) \gamma_S$}, is impossible under Thermal Operations, no matter $E$. In fact, $A(\sigma^\epsilon_S) > A(\ket{E}) =0$, and as we discussed $A$ is a Thermal Operation monotone. The transformation $\ket{E} \rightarrow \sigma^\epsilon_S$ is impossible despite the fact that every constraint based on a free energy measure is trivialised by adding enough (incoherent) work, that is $\Delta \tilde{F} (\ket{E}) > \Delta \tilde{F}(\sigma^\epsilon_S)$ for $E$ large enough. This suggests that something beyond a `generalised free energy' is needed to characterise thermodynamic transformations.

A more detailed understanding can be obtained noting that using `battery states' $\ket{E}$ as resources one can simulate any time-translation symmetric operation with Thermal Operations. Specifically, for any $U(1)$-covariant operation $\mathcal{E}$ we can find a battery state $\ket{E}$ such that 
\begin{equation*}
\tr{W}{\mathcal{T}(\rho_S \otimes \ketbra{E}{E}_W)} \approx_\epsilon \mathcal{E}(\rho_S)
\end{equation*} 
arbitrarily well (using Theorem~\ref{thm:stinespringcovariant}, we just need to use the battery to produce some diagonal state $\sigma_A$; that this can be done follows from an application of Theorem~\ref{thm:thermalnielsen}, see Appendix~B of Ref.~\cite{lostaglio2015quantum}). Conversely, no operation outside this set can be realised using Thermal Operations and battery states $\ket{E}$ (as it follows from Theorem~\ref{thm:stinespringcovariant}). Hence, the set of time-translation symmetric operations are all and only the channels that can be obtained with Thermal Operations and arbitrary energy states $\ket{E}$. Since adding energy states is exactly the construction that lifts all free energy measures, we see that asymmetry constraints are those that remain. 
These are the symmetry `backbone', describing constraints on the evolution of quantum coherence that follow from time-translation symmetry only. This identifies a crucial difference from the classical scenario, where all constraints are lifted by adding enough work. 
\begin{rmk}
Using this approach, Ref.~\cite{marvian2018coherence} studied the thermodynamic coherence costs of creating quantum states when work is available, showing that it is related to the quantum Fisher information defined in Example~\ref{ex:fisher}. It is also shown that coherence cannot be distilled in the form of a pure, uniform superposition even from an infinite number of full rank input states. 
\end{rmk}

There is an interesting decomposition of the quantum free energy, first derived in Ref.~\cite{janzing2006quantum}, that helps in understanding the previous discussion more concretely. 
Define $\Delta F_C(\rho_S) := \Delta F(\mathcal{D}(\rho_S))$ the \emph{classical free energy}. This can be seen to be equal to the (non-equilibrium) free energy of the vector of populations $x_i = \bra{E^S_i} \rho_S \ket{E^S_i}$, i.e. $\Delta F_C(\rho_S)  = F(\v{x}) - F(\v{g})$, where $F(\v{x}) = \sum_i x_i E^S_i - kT H(\v{x})$ and $H$ is the Shannon entropy. Then 
\begin{thm}[Free energy decomposition into incoherent and coherent parts \cite{janzing2006quantum, lostaglio2015description}]
	\begin{equation}
	\label{eq:freenergydecomposition}
	\Delta F(\rho_S) = \Delta F_C(\rho_S) + kT A(\rho_S),
	\end{equation}
	where $\Delta F_C(\rho_S)$ is the classical free energy and $A(\rho_S)$ is asymmetry with respect to time-translations,  defined in Eq.~\eqref{eq:asymmetry}. Under a Thermal Operation $\mathcal{T}$, $\Delta F_C(\mathcal{T}(\rho_S)) \leq \Delta F_C(\rho_S)$ \emph{and} $ A(\mathcal{T}(\rho_S)) \leq A(\rho_S)$.
\end{thm}

In other words the quantum (non-equilibrium) free energy additively decomposes in a component that is the free energy of the population only, measuring the distance of the population of $\rho_S$ from a thermal population; and a coherent component $kT A(\rho_S)$, measuring the distance between $\rho_S$ and the closest incoherent state. The latter interpretation is made precise from the fact that $A(\rho_S) = \min_{\sigma|\sigma = \mathcal{D}(\sigma_S)} S(\rho_S\|\sigma_S)$ (Proposition~2 of Ref.~\cite{gour2009relative}). Both components must independently decrease.
\begin{proof}[Proof of free energy decomposition]
	Using $\mathcal{D} = \mathcal{D}^\dag$ and $\gamma_S = \mathcal{D}(\gamma_S)$, we get $\tr{}{\rho_S \log \gamma_S} = \tr{}{\mathcal{D}(\rho_S) \log \gamma_S}$. Then, summing and subtracting $S(\mathcal{D}(\rho_S))$, we get
	\begin{equation}
	\Delta F(\rho_S) = kT (\tr{}{\rho_S \log \rho_S} - \tr{}{\mathcal{D}(\rho_S) \log \mathcal{D}(\rho_S)} + \tr{}{\mathcal{D}(\rho_S) \log \mathcal{D}(\rho_S)} - \tr{}{\mathcal{D}(\rho_S) \log \gamma_S})  
	\end{equation}
	The first two terms are $kT A(\rho_S)$ and the last two are $\Delta F(\mathcal{D}(\rho_S))$. As we have discussed before, $A(\mathcal{T}(\rho_S)) \leq A(\rho_S)$. Furthermore, by noting that $\mathcal{T}$ commutes with $\mathcal{D}$ (again due to symmetry), one can immediately verify using the contractivity of the relative entropy that $\Delta F_C (\mathcal{T}(\rho_S)) \leq \Delta F_C(\rho_S)$.
\end{proof}
From these considerations it is simple to see why the transformation $\ket{E} \rightarrow \ket{+}$ cannot happen under Thermal Operations: while for $E>0$ large enough certainly the classical free energy as well as the (total) quantum free energy are decreasing in the process, one has $A(\ket{+}) = \log 2$, $A(\ket{E}) = 0$, so the coherent component of the free energy would be increasing if the transition was possible. This immediately rules out the above as an allowed transformation. That the set of channels satisfying property~\eqref{eq:GPcondition} (Gibbs-preserving condition) but not property~\eqref{eq:covariancecondition} (symmetry condition) `outperform' Thermal Operations \cite{faist2015gibbs} should be intuitive from the above considerations, since they only have to decrease the total free energy.

We note in passing that both terms in the decomposition have an operational meaning: $\Delta F(\rho_S)$ is the maximum amount of work that can be extracted on average from the quantum state $\rho_S$ by applying general unitaries on system and a thermal environment;  while $\Delta F_C(\rho_S)$ is the maximum amount of work that can be extracted from $\rho_S$ on average by the same set of protocols after an energy measurement that destroys energy coherence (see, e.g., Ref.~\cite{kammerlander2016coherence} and the Appendix of Ref.~\cite{baumer2018fluctuating}). 
 
 \subsubsection{Application: work-locking and limits of semiclassical treatments}

Let us go back to the question of work extraction, discussed in Sec.~\ref{sec:workextractionincoherent} for incoherent states. One looks for a Thermal Operation of the form
\begin{equation}
\label{eq:workextractionprocess}
\mathcal{T}(\rho_S \otimes \ketbra{0}{0}_W) = \gamma_S \otimes \sigma_W,
\end{equation}
where $\sigma_W$ is some diagonal state that stores the work extracted from $\rho_S$ (for example, $\sigma_W = \ketbra{1}{1}$ with \mbox{$H_W = W \ketbra{1}{1}$} for deterministic work extraction). As discussed, since $\mathcal{T}$ are time-translation symmetric channels they commute with the dephasing operator $\mathcal{D}$. Applying a dephasing to both sides of Eq.~\eqref{eq:workextractionprocess} allow us to conclude that if the transformation in Eq.~\eqref{eq:workextractionprocess} is possible, also the following is possible:
\begin{equation}
\mathcal{T}(\mathcal{D}(\rho_S) \otimes \ketbra{0}{0}_W) = \gamma_S \otimes \sigma_W.
\end{equation}
In other words, the work that can be extracted from $\rho_S$ cannot exceed the work that can be extracted from $\mathcal{D}(\rho_S)$. This may seem a bit puzzling, since $\Delta F (\rho_S) > \Delta F(\mathcal{D}(\rho_S))$ for every state with coherence, and in particular the difference is exactly the coherent part of the non-equilibrium quantum free energy. The impossibility to convert the coherent part of the free energy into work with Thermal Operations was called \emph{work locking} in Ref.~\cite{lostaglio2015description}. There are interesting tradeoffs between the amount of coherence in the energy degenerate subspaces, which can be extracted as work, and coherence among distinct eigenspaces, which due to work locking cannot be extracted \cite{kwon2018clock}. 

To access the coherent part of the free energy, we need an external source of coherence. More specifically, we need to have at our disposal an ancillary system $R$, with Hamiltonian $H_R$ and in a state $\sigma_R$ with $[\sigma_R,H_R]\neq 0$, that aids the transformation. What $R$ does it to break the time-translation symmetry on $SW$. By inducing on $SW$ a channel that does not commute with $\mathcal{D}$ we circumvent work locking. 

$R$ is known as a quantum reference frame \cite{bartlett2007reference}; often this role is implicitly played by the classical field that, in standard treatments, is responsible for a generic unitary that one is allowed to apply on $S$, as in the discussion around Eq.~\eqref{eq:workextractionpassive}. Within this semiclassical approach, one simply posits that the change of average energy in the system during a unitary process is work; but for small scale thermodynamics this can be problematic, because it neglects the back-reaction of the system on the field, which may deteriorate. For example, one may assume that a unitary $U$ on the system and bath is realised as
\begin{equation}
\tr{R}{V(U) \rho_{S} \otimes \gamma_B \otimes \ketbra{\alpha}{\alpha}_R V(U)^\dag} \approx U \rho_{SB} U^\dag,
\end{equation}
where $V(U)$ is an energy-preserving unitary involving $SB$ and a field state represented, for example, by an optical coherent state $\ket{\alpha}_R$ (for simplicity, let us use $R$ also as a work storage system). 
A self-contained treatment accounts for the back-reaction on $\ket{\alpha}_R$, as well as the fact that unitaries on $SB$ can be realised only approximatively. To simply dismiss the problem by saying that the change of the state on $R$ is very small and can be neglected does not suffice: one could argue in the same way that the amount of extracted work is very small and can hence be neglected! Back-reactions can sum up over many uses, and in principle be large once a sizeable amount of overall work is extracted.

One can indeed approach the standard result $W_{\rm ave} \rightarrow \Delta F(\rho_S)$, when a very large but finite coherent source is at our disposal; but we only know of very specific interactions that are able prevent the deterioration of the field (which can only exist in infinite dimensional systems \cite{lostaglio2019coherence}), and protocols exploiting them require an energy investment whose rate becomes small only in the limit of a very large number of uses of the field. For further considerations on these issues, see e.g.  Ref.~\cite{klimovsky2013work, aberg2014catalytic, korzekwa2016extraction}.  
 
   \subsubsection{Modes of coherence and hierarchy of thermodynamic constraints}
   \label{subsection:modes}
   
   We discussed the role of the coherent properties of a quantum state in thermodynamics. In this respect, quantum states that are not symmetric admit a more refined decomposition in `chunks' that transform very simply under time translations. These are called modes of asymmetry for a general group $G$ \cite{marvian2014modes}, but here we will focus on $G=U(1)$, where they are called \emph{modes of coherence}: 

   \begin{defn}
   	Given $H_S$, construct the Bohr spectrum $\Omega$ defined as the set of all transition frequencies: \{$\omega \in \Omega \Leftrightarrow \exists E^S_i, E^S_j \in {\rm spec}(H_S)| \omega = E^S_i - E^S_j$\}, where ${\rm spec}(H_S)$ denotes the spectrum of $H_S$. If $\rho_S$ is a quantum state acting on $S$, it can be decomposed as
   	\begin{equation}
   	\rho_S = \sum_{\omega \in \Omega} \rho_S^{(\omega)},
   	\end{equation}
   	where each $\rho_S^{(\omega)}$ satisfies $\mathcal{U}_t(\rho_S^{(\omega)}) = e^{-i \omega t} \rho_S^{(\omega)}$ and is called \emph{mode of coherence} $\omega$.
   \end{defn}
\begin{ex}
	Let $H_S = \sum_n n E  \ketbra{n}{n}$. The modes of coherence of $\rho_S$ are given by $\omega = E\{\dots, -2,-1,0, 1,2,\dots\}$, with $
	\rho_S^{(\omega)} = \sum_{n=0}^{+\infty} \rho_{n+\omega,n} \ketbra{n+\omega}{n}$, where $\rho_{n,n+\omega}$ are the matrix elements of $\rho_S$ in the energy eigenbasis.
\end{ex}

\begin{rmk}
In most elementary examples it should be straightforward to identify the modes of coherence. There are however systematic ways of constructing them. In fact, finding them corresponds to decomposing $\rho_S$ according to a so-called irreducible tensor operator basis \cite{marvian2014modes}. In other words, $\rho_S = \sum_{\omega,m} \rho^{(\omega)}_m $ where $\mathcal{U}_g(\rho^{(\omega)}_m) = \sum_{m'} u^{(\omega)}_{mm'}(g) \rho^{(\omega)}_{m'}$ and $u^{(\omega)}_{mm'}(g)$ are the matrix elements of the irreducible representation of $U_g \otimes U^*_g$ labelled by $\omega$ (this relation can be understood by vectorisation). $\rho^{(\omega)}_m$ are known as an irreducible tensor operator basis. 
\end{rmk}

The relation between the assumption that the dynamics is symmetric under time translation and the modes of coherence is simple:  
\begin{thm}
	$\mathcal{T}$ is $G$-covariant if and only if 
	\begin{equation}
	\label{eq:modesintomodes}
	\mathcal{T}(\rho_S^{(\omega)}) = \mathcal{T}(\rho_S)^{(\omega)}, \quad \forall \rho_S, \forall \omega \in \Omega.
	\end{equation}
\end{thm}

\begin{proof}
	Let $P_\omega(\cdot) = \int dt e^{i \omega t} \mathcal{U}_t(\cdot)$. By direct computation it should be clear that $P_\omega$ is a projector on mode $\omega$, so that $P_\omega (\rho_S^{(\omega')}) = \delta_{\omega \omega'} \rhomega$. Then, assuming $\mathcal{T}$ is covariant,
	\begin{equation}
\mathcal{T}(\rho_S)^{(\omega)} = P_\omega \mathcal{T}(\rho_S) = \mathcal{T}(P_\omega(\rho_S)) = \mathcal{T}(\rho_S^{(\omega)}).
	\end{equation}
	Conversely, assume Equation~\eqref{eq:modesintomodes} holds. Then
	\begin{equation}
	\mathcal{U}_t \mathcal{T}(\rho_S) = \mathcal{U}_t \sum_{\omega} \mathcal{T}(\rho_S)^{(\omega)} =  \sum_{\omega} e^{-i \omega t} \mathcal{T}(\rho_S)^{(\omega)} = \mathcal{T}\left( \sum_{\omega}  e^{-i \omega t} \rhomega\right) = \mathcal{T}\left( \sum_{\omega}  \mathcal{U}_t (\rhomega)\right) = \mathcal{T}(\mathcal{U}_t (\rho_S)). 
	\end{equation} 
\end{proof}

We see now that we can separate the various constraints imposed by Thermal Operations as follows. Suppose there exist a Thermal Operation $\mathcal{T}$ such that $\mathcal{T}(\rho_S) = \sigma_S$. Then 
\begin{equation}
\label{eq:hierarchy}
\mathcal{T}(\rho_S^{(\omega)}) = \sigma_S^{(\omega)}, \quad \forall \omega \in \Omega.
\end{equation} 
The zero mode corresponds to the vector of population. Then, thanks to Theorem~\ref{th:thermalgibbs}, after an obvious correspondence between diagonal matrices and vectors of probabilities, the mode $\omega =0$ constraints corresponds to the existence of a Gibbs-stochastic matrix $G$ such that $G \rho_S^{(0)} = \sigma_S^{(0)}$, which is equivalent to 
\begin{equation}
\rho_S^{(0)} \succ_g \sigma_S^{(0)},
\end{equation}
i.e. thermo-majorisation. Hence, thermo-majorisation is a zero mode constraint of a hierarchy that also includes $\mathcal{T}(\rho_S^{(\omega)}) = \sigma_S^{(\omega)}$ for $\omega >0$, $\omega \in \Omega$.

\subsection{Thermodynamic constraints on the evolution of quantum coherence}

 \subsubsection{A general theorem connecting population and coherence constraints}
 
 So far we have considered the coherence constraints independently of the population dynamics, but it is clear that, if a quantum channel implements a given dynamics on the population, the corresponding coherent evolutions are limited by the overall complete positivity of the map.
Given some initial state $\rho_S$, for any given $x,y$ we are interested in 
\begin{align}
\max_{\E} & \; \; |\mathcal{E}(\rho_S)_{xy}| \\
\textrm{ subject to } & \E \circ \mathcal{U}_t = \mathcal{U}_t \circ \mathcal{E} \quad \forall t \\
  & \label{eq:classicalaction} \E(\ketbra{x}{x}) = \sum_{x'} P_{x'|x} \ketbra{x'}{x'},
\end{align}
for a stochastic matrix $P$ which will be later identified with the Gibbs-stochastic matrix that Thermal Operations induce on the population vector. 
We note in passing that the above can be written as a semidefinite program, using the channel-state duality and seeing the symmetry constraint $\E \circ \mathcal{U}_t = \mathcal{U}_t \circ \mathcal{E} $ as a projection of the space of quantum maps on the covariant subset (a `super $G$-twirling', see Eq.~(2.17) of Ref.~\cite{bartlett2007reference}). In physical terms, we can think of this problem as follows: if we know the classical action $P$ of $\E$, representing the energy flows induced by $\mathcal{E}$, how much coherence can be preserved? We will express the off-diagonal matrix elements of $\rho_S$ in terms of their magnitudes and phase factors as $\rho_{xy}=|\rho_{xy}|\vartheta_{xy}$. The symbol $\sum^{(\omega)}_{x,y}$ will indicate a sum over all indices $x$,$y$ such that $\omega_x - \omega_y = \omega$. Denoting $\sigma_S = \mathcal{E}(\rho_S)$, one has 
\begin{thm}[\cite{lostaglio2015quantum}, tightness conditions in~\cite{lostaglio2017markovian})]
	\label{thm:CP_bound}
	Let $\E$ be a time-translation symmetric map such that \mbox{$\sigma_S=\E(\rho_S)$}, and satisfying Eq.~\eqref{eq:classicalaction}. Then $|\sigma_{x'y'}|$ is bounded as
	\begin{equation}
	\label{eq:CP_bound}
	|\sigma_{x'y'}| \leq \sum^{\ompp{x}{y}}_{x,y} \sqrt{P_{x'|x}P_{y'|y}} |\rho_{xy}|, \quad \omega_{x'y'} := \omega_{x'}-\omega_{y'}.
	\end{equation}
\end{thm}
\begin{proof} We follow the proof given in Ref.~\cite{lostaglio2017markovian}.	The complete positivity of $\E$ is equivalent to the positivity of the Choi-Jamio\l kowski state $J[\E] := \E \otimes \mathcal{I} (\Phi_+)$, where $\Phi_+$ is the maximally entangled state ($\Phi_+ = \ketbra{\phi_+}{\phi_+}$, $\ket{\phi_+} \propto \sum_{i} \ket{ii}$) \cite{choi1975completely,bengtsson2006geometry}. However, $J[\mathcal{E}]$ satisfies \mbox{$e^{-i\widetilde{H} t} J[\mathcal{E}] e^{i\widetilde{H} t} = J[\mathcal{E}]$}, where \mbox{$\widetilde{H}= H_S \otimes \I -\I \otimes H_S^* $} (see Ref.~\cite{dariano2001optimal}, Eqs.~(18)-(19)). This can be seen as follows: 	for any unitary $U$, $\mathcal{U} \otimes \mathcal{U}^* (\Phi_+) = \Phi_+$ where $\mathcal{U}(\cdot) = U(\cdot)U^\dag$. 
	Taking $U= e^{-i H_S t}$ we obtain
	\be                                             
	[\mathcal{U}_t \circ \mathcal{E} \circ \mathcal{U}^\dag_t \otimes \mathcal{I}] (\Phi_+) = [\mathcal{U}_t \circ \mathcal{E} \otimes \mathcal{I}](\mathcal{I}\otimes \mathcal{U}^*_t)(\Phi_+) 
	= \mathcal{U}_t \otimes \mathcal{U}_t^* [\mathcal{E} \otimes \mathcal{I}](\Phi_+).
	\ee
	This immediately implies (using that $J$ is an isomorphism), $\mathcal{U}_t \circ \mathcal{E} \circ \mathcal{U}^\dag_t = \mathcal{E} \Leftrightarrow \mathcal{U}_t \otimes \mathcal{U}_t^* J[\mathcal{E}] = J[\mathcal{E}]$. 
	Hence, $\mathcal{E}$ is covariant if and only if $J[\mathcal{E}]$ is symmetric with respect to the Hamiltonian $\tilde{H}= H_S \otimes \I - \I \otimes H^*_S$.
	
	Hence, the Choi state is block diagonal in the eigenbasis of $\widetilde{H}$ and the positivity of $J[\E]$ is equivalent to positivity of each block. From the definition of the Choi state and denoting by $c^{x'|x}_{y'|y} = \bra{x'}\mathcal{E}(\ketbra{x}{y})\ket{y'}$ we get
	\be
	\label{eq:choi}
	J[\mathcal{E}]= \sum_{x,y} \sum^{\om{x}{y}}_{x',y'} c^{x'|x}_{y'|y}\ketbra{x'}{y'} \otimes \ketbra{x}{y}
	= \sum_{x',x} \sum^{\omp{x}{x}}_{y',y} c^{x'|x}_{y'|y}\ketbra{x'x}{y'y},
	\ee
	where we have rearranged the expression to emphasise the block-diagonal structure (we used $\omega_{x'y'} = \omega_{x y} \Leftrightarrow \omega_{x'x} = \omega_{y'y}$). Each block consists of matrix elements $c^{x'|x}_{y'|y}$ for which \mbox{$\omega_{x'}-\omega_x = \omega_{y'}-\omega_y=\omega$} and can thus be labelled by $\omega$ (see Fig.~\ref{fig:choi}). A necessary condition for the positivity of block $\omega$ is that for all $x,y$ and $x',y'$ within the block one has
	\begin{equation}
	\label{eq:opt_coh}
	|c^{x'|x}_{y'|y}|\leq \sqrt{P_{x'|x}P_{y'|y}}.
	\end{equation}
	Since $\sigma_{x'y'} = \sum^{(\omega_{x'y'})}_{x,y}c^{x'|x}_{y'|y}\rho_{xy}$, by the triangle inequality and Eq.~\eqref{eq:opt_coh} we obtain the result claimed in Eq.~\eqref{eq:CP_bound}.
\end{proof}

\begin{figure}[h]
	\begin{center}
		\includegraphics[width=0.5\textwidth]{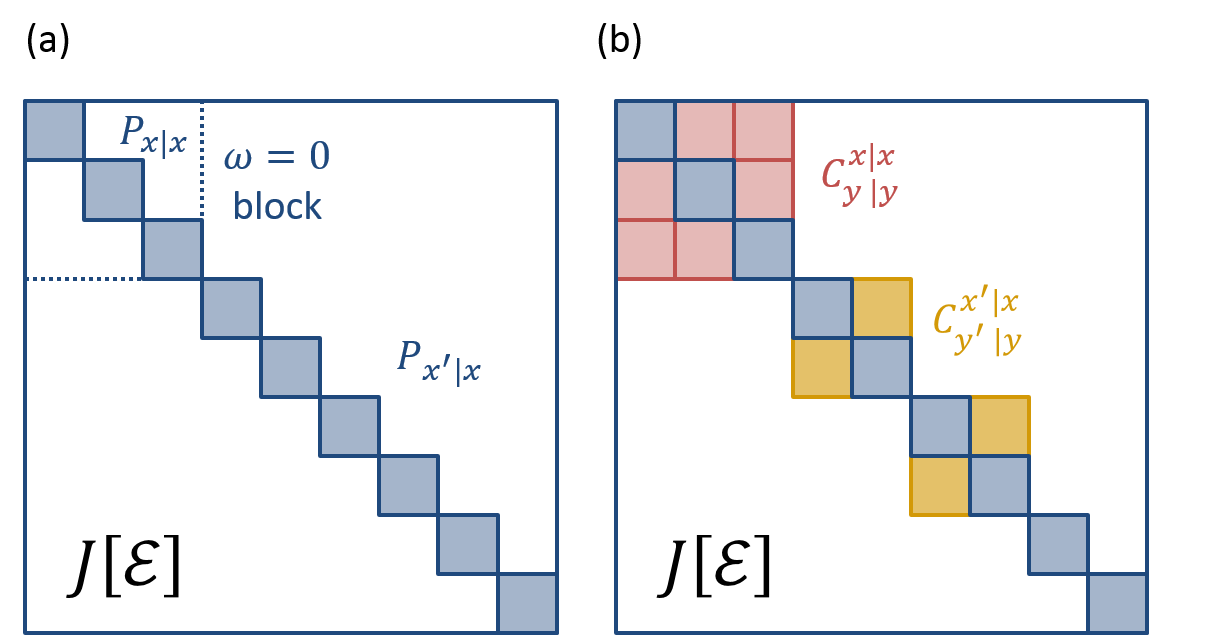}
	\end{center}
	\caption{Structure of the Choi matrix of a covariant channel. In (a) we emphasise that the diagonal elements of each block $\omega$ correspond to transition probabilities $P_{x'|x}$ with $\omega_{x'}- \omega_x = \omega$; in (b) we show that the off-diagonal elements of each block $\omega$ corresponds to transition amplitudes $c^{x'|x}_{y'|y}$ with $\omega_{x'}- \omega_{y'} = \omega_{x} - \omega_y = \omega$. Picture from Ref.~\cite{lostaglio2017markovian}.}
	\label{fig:choi}
\end{figure}

\subsubsection{Application: qubit Thermal Operations}

We now briefly survey two interesting applications of the previous theorem: the complete solution to the question of state transformations $\rho_S \rightarrow \sigma_S$ under Thermal Operations for qubit systems, and a result on the irreversibility of coherence transfers.

First, we can solve the theory of Thermal Operations for a single qubit. In these two-dimensional systems the most general Gibbs-stochastic matrix acting on the diagonal can be expressed as a function of a single parameter $\lambda$:
\begin{equation}
G = \begin{pmatrix}
1-\lambda e^{-\beta E} & \lambda \\
e^{-\beta E} \lambda & 1-\lambda \end{pmatrix}, \quad  \lambda \in [0,1].
\end{equation}
Define the initial and final state
\begin{equation}
\rho_S = \begin{pmatrix}
p & c \\
c & 1-p \end{pmatrix}, \quad \quad \sigma_S = \begin{pmatrix}
q & d \\
d & 1-q \end{pmatrix},
\end{equation}
where without loss of generality we can take $c,d \geq 0$, since an energy preserving unitary allows us to adjust the phase of the off-diagonal terms. 
 
The condition that $(p,1-p)$ is mapped into $(q,1-q)$ fixes $\lambda$, giving $\lambda = \frac{q-p}{g-p}g$.  
Theorem~\ref{eq:CP_bound} implies that $d \leq \sqrt{G_{0|0} G_{1|1}} c = \sqrt{(1- \lambda e^{-\beta E})(1-\lambda)}c$. This provides the final relation (see Fig.~\ref{fig:qubitthermal})
\begin{equation}
\label{eq:qubitboundcoherence}
d\leq \frac{\sqrt{(q(1-g)-g(1-p))(p(1-g)-g(1-q))}}{|p-g|}c.
\end{equation}
One can see that the bound is achievable by means of the Gibbs-preserving and time-translation symmetric channel (i.e., Thermal Process) \mbox{$\E(\cdot) = \sum_{\omega}K_{\omega} (\cdot) K_{\omega}^\dag$} with Kraus operators
 \begin{equation}
 K_0 = \sqrt{G_{0|0}} \ketbra{0}{0} + \sqrt{G_{1|1}} \ketbra{1}{1}, \quad K_1 = \sqrt{G_{1|0}} \ketbra{1}{0}, \quad K_{-1} = \sqrt{G_{0|1}} \ketbra{0}{1},
 \end{equation}
 fixed by the above choice of $\lambda$.
 One can directly check that $\mathcal{E}(\gamma_S) = \gamma_S$ and $\mathcal{E} \circ \mathcal{U}_t = \mathcal{U}_t \circ \mathcal{E}$ (covariance also follows immediately from Proposition~7 of Ref.~\cite{marvian2016how}). A direct calculation shows that this channel saturates the bound of Eq.~\eqref{eq:qubitboundcoherence}. Any other state `inside the boundary' can be achieved by this optimal channel followed by a partial dephasing $\mathcal{D}_s = (1-s)\mathcal{I} + s \mathcal{D}$, where $s \in [0,1]$ and $\mathcal{I}$ is the identity channel ($\mathcal{D}_s$ is a Thermal Operation). 
 
 In fact, it was proved in Ref.~\cite{cwiklinski2015limitations} that this transformation can be achieved by a Thermal Operation.
 
 \begin{figure}[h]
 	\begin{center}
 		\includegraphics[width=0.4\textwidth]{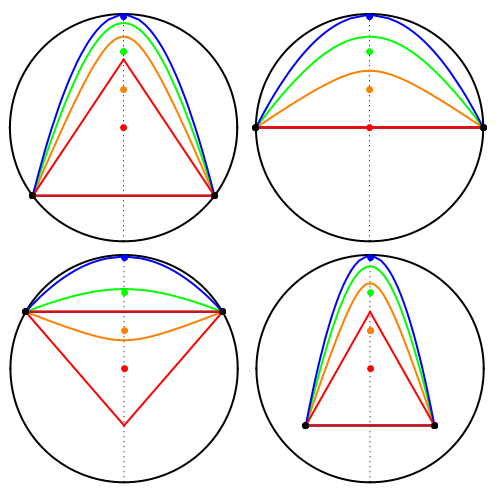}
 	\end{center}
 	\caption{Boundaries of the set of qubit states achievable by means of Thermal Operations for various choices of temperatures and initial states. The regions are represented on the $xz$ plane of the Bloch sphere due to symmetry under rotations about $z$ (the state $\ket{0}$ is on top and the state $\ket{1}$ at the bottom). Initial states are the black dots at the boundary of the various regions. The red/orange/green/blue dots on the $z$ axis represent the thermal state for different choices of the temperature $T$, from $T=\infty$ (red) to $T=0$ (blue). The corresponding boundaries represent the extremal states achievable from the initial state with a Thermal Operation at that temperature. Note in passing that the region is a polytope only for $T=\infty$. Figure from Ref.~\cite{lostaglio2015quantum}.}
 	\label{fig:qubitthermal}
 \end{figure}

 \subsubsection{Application: irreversibility in coherence transfers}
 We have seen that coherence transformations have a mode structure, but it is interesting to analyse in more detail how the evolution of coherence inside a given mode happens. We do this only for a very simple example \cite{lostaglio2015quantum}.
 
 Let us assume that $\rho_S$, with Hamiltonian $H_S = \sum_{i=0}^2 E^S_i \ketbra{E^S_i}{E^S_i}$, has a single non zero coherence element $|\rho_{01}|>0$, i.e. a superposition of energies $E^S_0$ and $E^S_1$ (and any population vector). We wish to transport this coherence `up in energy' to a superposition of the energies $E^S_1$ and $E^S_2$, where $\Delta E = E^S_2 - E^S_1 = E^S_1 - E^S_0$; in other words, we want the final state $\sigma_S$ to have the largest $|\sigma_{12}|$ possible, and we do not care about the final population.  What are the limits imposed by Thermal Operations, and how do they compare to the reverse process of transporting coherence `down in energy' $12 \rightarrow 01$?
 \begin{equation}
 \rho_S = \begin{pmatrix}
 ? & |\rho_{01}| & ? \\
 |\rho_{01}|  & ? & 0\\
 ? & 0 & ?
 \end{pmatrix} \quad \longleftrightarrow \quad  \sigma_S = \begin{pmatrix}
 ? & 0 & ? \\
 0  & ? & |\sigma_{12}|\\
 ? & |\sigma_{12}| & ?
 \end{pmatrix}. 
 \end{equation}
 
  Recall that one has the bound
 \begin{equation}
 	\label{eq:boundcovarianttrit}
 |\sigma_{12}| \leq \sqrt{G_{1|0} G_{2|1}} |\rho_{01}|.
 \end{equation}
 From the Gibbs-preserving condition, denoting the thermal vector of the system by $\v{g} = (g_0,g_1,g_2)$,
 \begin{equation}
 G_{1|0}g_0 + G_{1|1} g_1 + G_{1|2} g_2 = g_1 \Rightarrow G_{1|1} = 1- G_{1|0}g_0/g_1 - G_{1|2}g_2/g_1
 \end{equation}
 Since $G_{1|1} \geq 0$, this gives $
 G_{1|0} \leq g_1/g_0 - G_{1|2}g_2/g_0 \leq g_1/g_0 = e^{-\beta \Delta E}$.
 In fact, the same reasoning yields \mbox{$G_{i|j} \leq e^{-\beta (E^S_i-E^S_j )}$}. By substitution in Eq.~\eqref{eq:boundcovarianttrit} we get
 \begin{equation}
 \label{eq:finaltritbound}
 |\sigma_{12}| \leq e^{-\beta \Delta E} |\rho_{01}|.
 \end{equation}
  Taking coherence ``up in energy'' can be done, but an exponential amount is lost. The bound is achievable and, furthermore, the reverse process ($12 \rightarrow 01$) can be done perfectly. To see this, take a single bosonic mode in a thermal state, $\gamma_B=\frac{1}{Z}\sum_{n=0}^{\infty}e^{-\beta n \Delta E}\ketbra{n}{n}$,
 where $Z=(1-e^{-\beta \Delta E})^{-1}$. Consider now the energy preserving unitary on system and bath given by
 \begin{equation*}
 	U=\ketbra{00}{00}+\ketbra{01}{10}+\ketbra{10}{01}
 	+\sum_{i=2}^{\infty}\ketbra{1;i-1}{2;i-2}+\ketbra{0;i}{1;i-1}+\ketbra{2;i-2}{0;i}.
 \end{equation*}
 It is a direct calculation to show
 \begin{align*}
 	\tr{B}{U(\ketbra{2}{1}\otimes\gamma_B)U^{\dagger}}=\ketbra{1}{0}, \quad \quad  & \textrm{(perfect transport of coherence down in energy)}\\
 	\tr{B}{U^{\dagger}(\ketbra{1}{0}\otimes\gamma_B)U}=e^{-\beta \Delta E}\ketbra{2}{1}, \quad & \textrm{(exponentially damped transport of coherence up in energy)}
 \end{align*}
 This illustrates how the irreversibility of energy transfers under Thermal Operations is reflected in the irreversibility of coherence transfers within a mode .
 
 \section*{Concluding remarks}
 
 There are many other results in the resource theory of Thermal Operations that I either only touched upon very briefly, or I did not discuss at all. However, you should now have the necessary background to explore the most recent developments. An incomplete list includes: the low temperature regime and the third law \cite{narasimhachar2015low,wilming2017third, masanes2017general}, interpolations between the single shot and average \cite{aberg2013truly, richens2016work} as well as single-copy and collective \cite{llobet2019collective} work extraction regimes, fluctuation theorems in the resource theory context \cite{aberg2016fully, alhambra2016fluctuating2, holmes2018coherent, boes2019passing}, relating the resource theory framework to the axiomatic approach of Lieb and Yngvason \cite{lieb2013entropy, weilenmann2015axiomatic}, correlations in the single-shot regime \cite{lostaglio2015stochastic, mueller2017correlating, sapienza2019correlations} (including quantum coherence \cite{lostaglio2019coherence, marvian2019no}) limits to catalysis \cite{ng2015limits},  explorations beyond i.i.d. limits \cite{chubb2017beyond, korzekwa2019avoiding}, i.i.d in many-body systems \cite{sagawa2019asymptotic}, approximate transformations~\cite{renes2016relative, horodecki2017approximate, van2017smoothed}, finite heat capacities \cite{richens2017finite}, conditioned Thermal Operations \cite{narasimhachar2017resource}, Gaussian Thermal Operations~\cite{serafini2019gaussian, narasimhachar2019thermodynamic}, study of the role of batteries \cite{lipka2019second}, thermodynamics of quantum channels~\cite{faist2019thermodynamic}, optimal cooling protocols~\cite{alhambra2018heat} and many more directions more or less tightly related to the framework described here (see, e.g., references in \cite{goold2016role, vinjanampathy2016quantum, ng2018resource, bera2019thermodynamics}). I hope this introduction will help you navigate the growing literature of this subject, develop new connections with complementary approaches, find practical applications to the framework and identify genuinely quantum effects in quantum thermodynamics.

\bigskip

 {\bf Acknowledgments.} I am particularly indebted to Kamil Korzekwa, Antony Milne, Terry Rudolph for many discussions on these topics during my PhD, and in particular David Jennings and Raam Uzdin for their extensive comments. Many thanks to Antonio Acin, Alessio Belenchia, Dario Egloff, Chung-Yun Hsieh, Mohammad Mehboudi, Markus M\"uller, Marti' Perarnau-Llobet, Valerio Scarani, Ivan \u{S}upi\'c for useful discussions and helpful comments on an earlier draft, and Iman Marvian for the argument in Remark~4. I acknowledge financial support from the the European Union's Marie Sk\l odowska-Curie individual Fellowships (H2020-MSCA-IF-2017, GA794842), Spanish MINECO (Severo Ochoa SEV-2015-0522 and project QIBEQI FIS2016-80773-P), Fundacio Cellex and Generalitat de Catalunya (CERCA Programme and SGR 1381).
 \bibliography{Bibliography_thermodynamics}
 
 \appendix
 
 \section*{Appendix: Proof of Hardy-Littlewood-Polya theorem (Theorem~\ref{thm:majorisationbistochastic})}

 \label{appendix:proofhlp}
\begin{proof} If $\v{x}, \v{y} \in \R^2$, then $\v{x} \succ \v{y}$ if and only if $x^{\downarrow}_1 \geq y^{\downarrow}_1$ and \mbox{$x^{\downarrow}_1 + x^{\downarrow}_2 = y^{\downarrow}_1+ y^{\downarrow}_2$}. This implies $x^\downarrow_2 \leq y^\downarrow_2$, so $x^\downarrow_1 \geq y^\downarrow_1 \geq y^\downarrow_2 \geq x^\downarrow_2$. Hence, $y^\downarrow_1 = t x^\downarrow_1 + (1-t)x^\downarrow_2$ for some $t \in [0,1]$. From this and $y^\downarrow_2 = x^\downarrow_1 + x^\downarrow_2 - y^\downarrow_1$, we get $y^\downarrow_2 = (1-t)x^\downarrow_1 + t x^\downarrow_2$. Hence $\v{y}$ can be obtained from $\v{x}$ by means of a doubly stochastic matrix. We now proceed by induction. Assume the case $n-1$. By means of permutations, assume both $\v{x}$ and $\v{y}$ are sorted in decreasing order. Since $\v{x} \succ \v{y}$, we have $x_1 \geq y_1 \geq x_n$ (the second inequality follows from $x_n \leq y_n \leq y_1$). Let $k$ be the smallest index such that $x_1 \geq y_1 \geq x_k$. Then $y_1 = t x_1 + (1-t)x_k$, $t \in [0,1]$. Let $\v{z} = T_1 \v{x}$, where $T_1$ is a doubly-stochastic matrix such that $T_1 x_1 =  t x_1 + (1-t)x_k$ and $T_1 x_k = (1-t)x_1 + t x_k$ ($T_1$ it acts trivially on any other $x_j$). Note that $z_1 = y_1$. Moreover, denote by $\tilde{\v{z}}$ and $\tilde{\v{y}}$ the vectors $\v{z}$ and $\v{y}$ truncated of the first element. We have 
 \be
 \tilde{\v{z}} = (x_2,...,x_{k-1},(1-t)x_1 + t x_k, x_{k+1}, ...,x_n).
 \ee
 By definition of $k$, $x_1 \geq ... \geq x_{k-1} \geq y_1 \geq ... \geq y_n$. It follows that $\sum_{i=2}^m x_i \geq \sum_{j=2}^m y_i$ for all $m=2,..,k-1$. For $m \geq k$, due to $\v{x} \succ \v{y}$,
 {\small 
 	\be
 	\nonumber
 	\sum_{i=2}^m \tilde{z}_i = \sum_{i=2}^{k-1} x_i + (1-t)x_1 + tx_k + \sum_{i=k+1}^m x_i = \sum_{i=1}^m x_i - tx_1 + (t-1)x_k \geq  \sum_{i=1}^m y_i - y_1 = \sum_{i=2}^m \tilde{y}_i.
 	\ee
 }
 Equality holds when $m=n$ because $\v{x} \succ \v{y}$. We conclude that $\tilde{\v{z}} \succ \tilde{\v{y}}$. By induction hypothesis, there is a set of doubly-stochastic matrices $T_2,...,T_p$ (each acting non trivially only on two elements of $\tilde{\v{z}}$) such that $T_p \dots T_2 \tilde{\v{z}} = \tilde{\v{y}}$. Hence, $T_p \dots T_2 T_1 \v{x} = \v{y}$. Each $T$ 
 is a convex combination of the identity and a transposition. Hence, the composition of the $T_i$ is a convex combination of permutations. A convex combination of permutations is a doubly stochastic matrix, so we conclude.	
 
 Conversely, without loss of generality, assume $\{x_i\}$ are sorted in non-increasing order. By assumption $y_j = \sum_{i=1}^n B_{j|i} x_i$, with $B$ doubly-stochastic. Then $\sum_{j=1}^k y_j = \sum_{i=1}^n t_i x_i$, where we defined $t_i = \sum_{j=1}^k B_{j|i} \in [0,1]$. Moreover, $\sum_{i=1}^n t_i = k$. Then,
 \begin{eqnarray*}
 	\sum_{j=1}^k y_j - \sum_{j=1}^k x_j =	\sum_{i=1}^n t_ i x_i - \sum_{i=1}^k x_i &=& 	\sum_{i=1}^k (t_ i-1) x_i  + \sum_{i=k+1}^n t_i x_i  + \left(k-\sum_{i=1}^n t_i\right)x_k= \\
 	\sum_{i=1}^k (t_ i-1) (x_i - x_k ) &+& \sum_{i=k+1}^n t_i (x_i - x_k) \leq 0.
 \end{eqnarray*}
 So $\sum_{j=1}^k y_j \leq \sum_{j=1}^k x_j$, and equality holds for $k=n$ because $B$ is doubly-stochastic. We conclude that $\v{x} \succ \v{y}$.  
 \end{proof}
 
 	Reconsidering the previous proof, one can note that we proved the equivalence of $\v{x} \succ \v{y}$ and the existence of a doubly-stochastic map from $\v{x}$ to $\v{y}$. However, two more equivalent conditions can be deduced:
 	\begin{enumerate}
 		\item $\v{y}$ is in the convex hull of the permutations of $\v{x}$,
 		\item $\v{y}$ can be obtained from $\v{x}$ by means of a sequence of doubly-stochastic matrices which have the property that each acts non trivially only on a 2-level subsystem (technically these are known as $T$-transforms).
 	\end{enumerate}
\end{document}